\documentclass[10pt,twocolumn]{article}

\usepackage{amsmath}
\usepackage{amsthm}
\usepackage{appendix}
\usepackage{array,multirow}
\usepackage{caption}
\usepackage{cite}
\usepackage{comment}
\usepackage[margin=1.0in]{geometry}
\usepackage{fancyhdr}
\usepackage{float}
\usepackage{xcolor}
\usepackage{soul}
\usepackage{stix}
\usepackage{graphicx}
\usepackage{epigraph}
\usepackage{wrapfig}
\usepackage{cancel}
\usepackage{makecell}
\usepackage{hyperref}
\usepackage{stmaryrd}
\usepackage{subcaption}
\usepackage{listings}
\usepackage{cleveref}

\usepackage{mathtools}

\graphicspath{{./Figures/}}

\newtheorem{theorem}{Theorem}[section]

\newtheorem{lemma}[theorem]{Lemma}
\theoremstyle{definition}
\newtheorem{definition}[theorem]{Definition}

\newcommand{\Fkin}{kinematic indicator }
\newcommand{\FkinUH}{F_{\textnormal{H}}}
\newcommand{\TMSlit}{\breve{T}M}

\newcommand{\EH}{E_{\textnormal{H}}}
\newcommand{\ELab}[1]{E_{#1}}

\crefname{appendixtheorem}{lemma}{lemmas}

\date{\today}

\makeatletter
\def\endthebibliography{%
\def\@noitemerr{\@latex@warning{Empty `thebibliography' environment}}%
\endlist
}
\makeatother

\title{
The Vlasov Bivector: 
A Parameter-Free Approach to Vlasov Kinematics
}

\author{Finlay Gunneberg${}^{1,2,}$\footnote{email: f.gunneberg1@lancaster.ac.uk \\ ORCID:  https://orcid.org/0009-0009-1225-6912}, Jonathan Gratus${}^{1,2,}$\footnote{email: j.gratus@lancaster.ac.uk \\ ORCID:  https://orcid.org/0000-0003-1597-6084}, and Harvey Stanfield${}^{3,}$\footnote{email: harvey.stanfield@postgrad.manchester.ac.uk  \\ ORCID:  https://orcid.org/0000-0001-7197-3983}}

\begin{document}

\maketitle

${}^1$ Department of Physics, Lancaster University, Lancaster, LA1 4YB, UK

${}^2$ The Cockcroft Institute, Daresbury Labs, Keckwick Ln, Daresbury, WA4 4AD, UK

${}^3$ Department of Physics \& Astronomy, Manchester University, Manchester, M13 9PL, UK

Emails and ORCIDS

\begin{abstract}
Plasma kinetics, for both flat and curved spacetime, is conventionally performed on the mass shell, a 7--dimensional time-phase space with a Vlasov vector field, also known as the Liouville vector field. The choice of this time-phase space encodes the parameterisation of the underling 2nd order ordinary differential equations. By replacing the Vlasov vector on time-phase space with a bivector on an 8--dimensional sub-bundle of the tangent bundle, we create a parameterisation free version of Vlasov theory. This has a number of advantages, which include working for lightlike and ultra-relativistic particles, non metric connections, and metric-free and premetric theories. It also works for theories where no time-phase space can exist for topological topological reasons. An example of this is when we wish to consider all geodesics, including spacelike geodesics. 

We extend the particle density function to a 6--form on the subbundle of the tangent space, and define the transport equations, which correspond to the Vlasov equation. We then show how to define the corresponding 3--current on spacetime. We discuss the stress-energy tensor needed for the Einstein-Vlasov system.

This theory can be generalised to create parameterisation invariant Vlasov theories for many 2nd order theories, on arbitrary manifolds. The relationship to sprays and semi-sprays is given and examples from Finsler geometry are also given.
\end{abstract}


\section{Introduction}
In Vlasov and Boltzmann theories of kinematics, one is interested in the dynamics of a scalar field over a time-phase space which represents a particle density.
The fields are typically functions of time, 3 positional coordinates, and 3 velocity or momentum coordinates.
The dynamic equations are written in terms of a first order operator on this scalar field.
In the case of the Vlasov equation, the action of the operator on the scalar field is zero, whereas in the case of the Boltzmann equation, the right hand side is nonzero to account of collisions.
This 7-dimensional time-phase space, which we call the \textit{kinematic domain}, is a subspace of the 8-dimensional tangent bundle. It corresponds to the chosen parameterisation one uses for the underlying geodesics of the chosen Lorentz force equation. 
Sometimes there is a natural choice for this 7-dimensional space, such as the mass shell (see \cite{andreasson_einstein-vlasov_2011}).
In other cases it has to be chosen; for example, in the case of null geodesics.
Furthermore, there are even cases where it is not even possible to construct a kinematic domain;
for example, in the case when we consider all geodesics, including spacelike geodesics.

\begin{figure}[tb]
    \centering
    \includegraphics[width=0.4\linewidth,trim={3cm 4cm 3cm 4cm},clip]{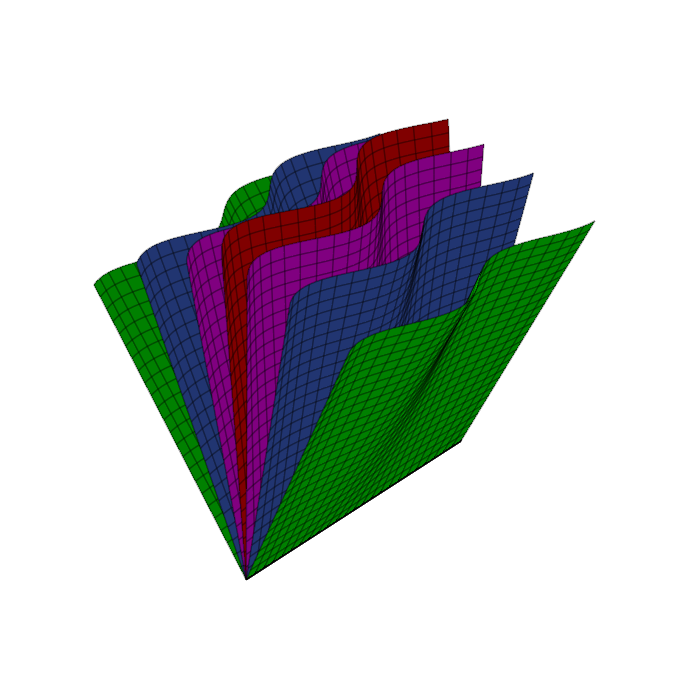}
    \caption{The Vlasov Bivector integrates to form leaves, like the leaves of a book. The density of the leaves also represent the particle distribution.}
    \label{fig:intro Leaves}
\end{figure}

In this article we investigate an alternative approach which does not require a choice of kinematic domain.
We work on an 8-dimensional conic bundle which is subset of the tangent bundle, and we replace the Vlasov vector, with the Vlasov bivector. 
This correspond to not choosing a parameterisation for the underlying ordinary differential equations (ODEs) of the system. 
The integral 2--dimensional surfaces of this bivector are depicted in figure \ref{fig:intro Leaves}. 
There are many advantages to doing this which we discuss in the following subsection.

Although we concentrate on the geodesic and Lorentz force equations and the corresponding Vlasov fields, this approach can be applied to the kinematics for any 2nd order ODEs, on any arbitrary dimensional base manifold.

In this article we first summarize the standard kinematic domain approach, and make the link with sprays and semi-sprays. 
We then define the Vlasov bivector and particle density 6--form, and give the equations of motion for the latter. 
We show how to go between the approaches, when the kinematic domain approach exists. We also show how to calculate the current 3--form the stress-energy tensor needed for the Einstein-Vlasov system.

In this article, we use sprays and Finsler geometry as examples to put this work in context. However, for the reader unfamiliar with these concepts, all statements about sprays, semi-sprays and Finsler geometry may be safely ignored. Details linking our work and sprays is given in \cref{Spray Section}.

\subsection{The Standard Vlasov Approach}

\begin{figure}[tb]
\begin{subfigure}{0.45\textwidth}
\centering
\includegraphics[scale=0.45]{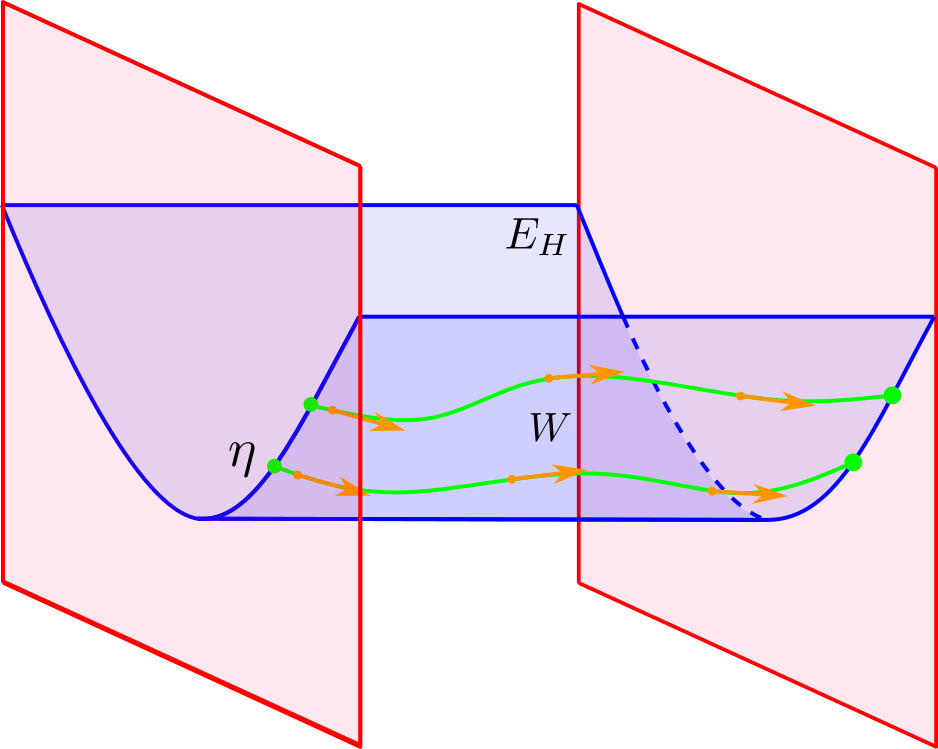}\\
\caption{Prolongations for a metric compatible connection.}
\label{Fig_Mass shell slip compat}
\end{subfigure}
\quad
\begin{subfigure}{0.45\textwidth}
\centering
\includegraphics[scale=0.45]{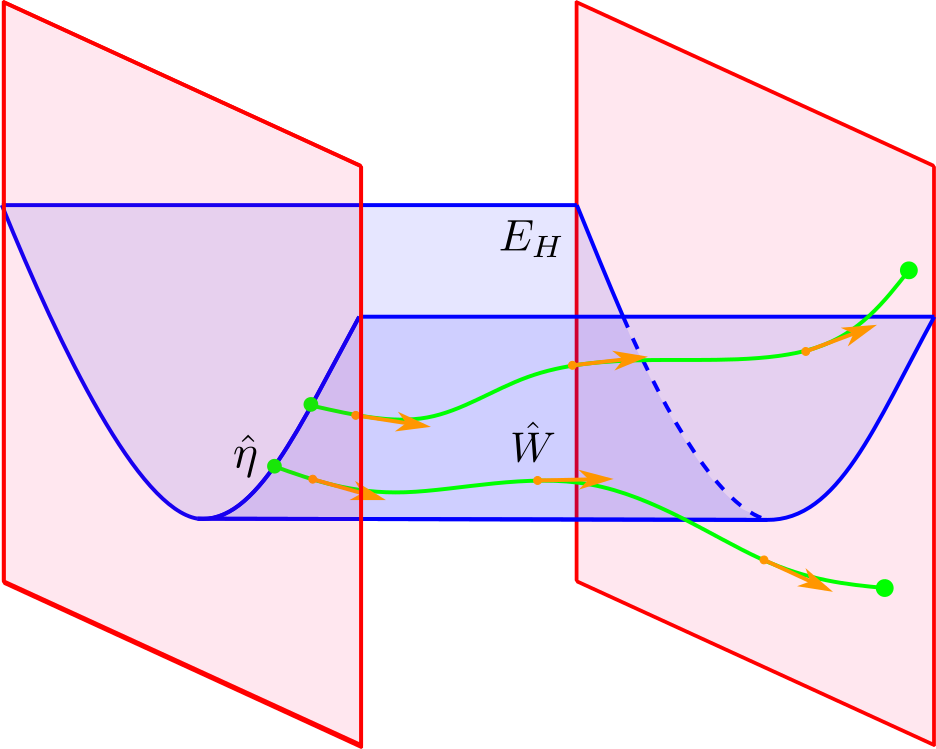}\\
\caption{Prolongations for a non-metric compatible connection.}
\label{Fig_Mass shell slip non compat}
\end{subfigure}
\caption{\label{Fig Mass Shell Slip} 
Illustration of integral curves, i.e the prolongations, in the mass shell for the case of a Vlasov field built from a force equation with a metric compatible connection (\cref{Fig_Mass shell slip compat}) and a non metric compatible connection (\cref{Fig_Mass shell slip non compat}). 
The mass shell is represented by the blue sheet, integral curves $\eta,\;\hat{\eta}$ of the Vlasov fields $W,\hat{W}$ are given by the green lines. 
The Vlasov fields themselves are depicted by the orange arrows.
The green curves show are initially on the unit hyperboloid but do not remain on it for the case where the Vlasov field is built from a non metric compatible connection.}
\end{figure}

When performing Vlasov kinematics, a 7--dimensional time-phase space $E$ is chosen upon which to construct the particle density scalar field $f_E$. 
We will refer to $E$ as the kinematic domain, and it is a bundle over the spacetime manifold $M$.
The first order differential equation can be represented by the action of a vector field called the Vlasov vector field $W_E\in\Gamma TE$,
\begin{equation}\label{Intr Vlasov Eqn on E}
W_E\langle f_E\rangle =0,
\end{equation}
where the angled brackets $\langle\bullet\rangle$ indicate the action of a vector field on a scalar field.
In this article we frequently use the Vlasov field as an example, both for charged and uncharged particles.

We refer to $W_E$ as the Vlasov vector field, or simply the Vlasov field.
It is also referred to as the Liouville vector field in some literature.
Throughout this paper/article we use the language of differential geometry on the tangent bundle to describe Vlasov systems.
For an overview of this technique see \cite{sarbach_geometry_2014}.
For an example of the benefits of using geometric techniques for Vlasov and related systems, see \cite{glinsky_coordinate-free_2024}.
The solutions to the underlying second order ODEs are called trajectories and are the worldlines of the particles. 
The tangent vectors to the trajectories, that is the velocities of particles, are curves in the kinematic domain, called the prolongation of the trajectories. 
Prolongations are integral curves of the Vlasov vector field. 
Although we primarily refer to the Vlasov field throughout this paper, these ideas can be applied to any kinematic theory in which the trajectories are governed by second order ODEs.

\begin{figure}[tb]
\centering
\textbf{}\includegraphics[scale=0.8]{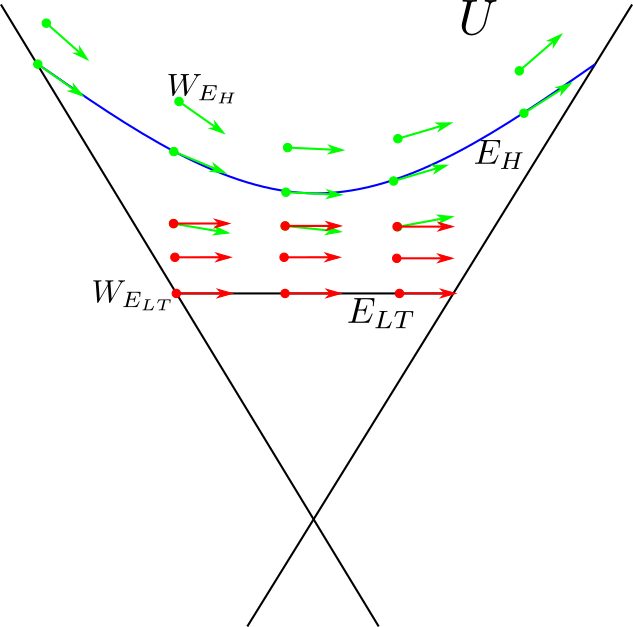}
\caption{\label{Fig Hyperboloid and Lab Bundles} The 7-dimensional kinematic domains $\EH$ and $\ELab{t}$ are given by slices of the 8-dimensional conic bundle $U$.
The unit hyperboloid $\EH$ is given by the dark blue hyperbola and a lab time bundle $\ELab{t}$ by the horizontal black line.
The vector fields $W_{\EH}$ and $W_{\ELab{t}}$ (green and red arrows respectively) are tangent to their respective kinematic domains.}
\end{figure}

The kinematic domain $E\subset TM$ is a sub-bundle of the tangent bundle over the base manifold $M$ and there are many choices for this sub-bundle.
For example, if $M$ has a spacetime metric, $g$ (with signature $(-,+,+,+)$), then one choice is to use the unit mass shell $\EH$, also known as the upper unit hyperboloid:
 \begin{equation}
 \EH=\{ \underline{v}\in TM \colon g(\underline{v},\underline{v})=-1, \mbox{ $\underline{v}$ is future pointing} \}.
\label{Eqn Intro EH}
 \end{equation}
This corresponds to proper-time parameterisation and is the natural choice for the relativistic Vlasov-Maxwell system and the Einstein Vlasov system \cite{andreasson_einstein-vlasov_2011}.
When performing plasma kinematics with the Vlasov-Maxwell system, we typically formulate the Lorentz force equation in terms of a metric compatible connection. One disadvantage of the upper unit hyperboloid is that prolongations of trajectories for a non-metric compatible connection will not, in general, remain on the mass shell (see \cref{Lem Non Metric Vlasov} for an example).
This is visualised in \cref{Fig Mass Shell Slip}. 
 
 Another choice of kinematic domain is the lab time bundle.
 This is characterised by a lab time scalar field $t\in\Gamma\Lambda^0M$, where $\underline{v} \langle t \rangle >0$ if $\underline{v}$ is future pointing, and is defined by
 \begin{equation}
 \ELab{t}=\{ \underline{v}\in TM \colon \underline{v}\langle t \rangle =1\text{ and $\underline{v}$ is timelike} \}.
 \label{Eqn Intro ELab}
 \end{equation}
 Working on this bundle adds an additional term to the Vlasov equation relative to the Vlasov field on $\EH$.
 Unlike the unit hyperboloid it does not require a metric compatible connection.
It is also useful when other quantities are defined with reference to a global time (see e.g. \cite{warwick_moment_2023}).
The disadvantage of this approach is that one has to choose a lab time coordinate and there may be complicated transformations from one lab time coordinate to another.
An illustration of kinematic domains as ``slices'' of a conic bundle is given in \cref{Fig Hyperboloid and Lab Bundles}.

Given two kinematic domains, say $E$ and $\hat{E}$, it is possible to transform the Vlasov field $W_E\in\Gamma TE$ to another Vlasov field $W_{\hat{E}}\in\Gamma T\hat{E}$ in such a way that the trajectories in the base space are unaffected.
This transformation corresponds to a reparameterisation of the trajectories associated with the Vlasov field.
For example, the geodesic equation on the unit hyperboloid $\EH$ becomes the pre-geodesic equation on the lab time bundle $\ELab{t}$, and by analogy the Lorentz force equation becomes the pre-Lorentz force equation. 
This extra force term, which is always proportional to the velocity, gives rise to an extra term in the Vlasov field.

\subsection{The Advantages of the Parameter Free Approach}
\label{Sec: Intro Advantages}

The primary goal of this work is to present the formalism of the Vlasov field and the corresponding Vlasov equation, in a way which does not require choosing a kinematic domain, or equivalently, a parameterisation. 
It is a considerable abstraction to pass from the 7-dimensional kinematic formalism to the 8-dimensional formalism, and several new concepts must be defined, such as the Vlasov bivector and the particle density form. 
We identify here a number of reasons why this new formalism is justified.

In this formalism we work on the conic bundle $U\subset TM$, which is a bundle over $M$ and has the same dimension as $TM$. 
This is a subset of the slit tangent bundle $\TMSlit=TM\backslash\{0\}$, which is the subset of $TM$ which excludes the zero vectors. 
It has the conic properties, namely if $\underline{u}\in U$ then for any $\lambda\neq 0$ we have $\lambda\underline{u}\in U$. 
Thus is the set of all possible velocities, which allows for rescaling. 
We say that $U$ is time orientable if it the disjoint union $U=U^+\cup U^-$ where $U^+$ are future pointing. 

There are multiple advantages to this approach:
\begin{itemize}
    \item It removes the arbitrariness of the kinematic domain and the need for additional terms. 
    Thus this approach is fundamentally free of the choice of parameterisation.

    \item It works when modelling lightlike particles where the kinematic domain is 6-dimensional, and hence one cannot use the hyperboloid bundle, $\EH$.
    One could use a lab time bundle $E=\{ \underline{v}\in TM \colon \underline{v}\langle t \rangle =1\text{ and } g(\underline{v},\underline{v})=0\}$.
    In the case of Minkowski spacetime, with $t$ given by the usual Minkowski coordinate, the prolongations of lightlike geodesics remain on $E$. 
    However, in general where there is gravity or where we use an arbitrary time scalar field, this is not the case, and the prolongations of lightlike geodesics will not remain on $E$. See appendix \cref{Lem Null Geo} for an example of this. 
    Thus one has two choices, either to use our approach given here or to use the pre-geodesic equation.
    
    \item In particle accelerators charged particles travel at speeds very close to the speed of light and one can use an ultra-relativistic approximation. 
    Although massless particles are not affected by an electromagnetic fields, ultra-relativistic particles do respond in the limit where the electromagnetic field becomes infinite. 
    In \cite{burton_asymptotic_2007}, the authors consider the ultra-relativistic approximation for a charged fluid. 
    The approach given here will enable the ultra-relativist approximation for a kinetic description of charged particles. 

   \item As seen in \cref{Fig Mass Shell Slip}, the prolongations of solutions to the Lorentz force equation  do not, in general, remain on $\EH$ if the connection is not metric compatible. 
   An example is given in the appendix. 
   By contrast, our approach works with force equations constructed from any connection, and more general models of acceleration \cite{torrome_effect_2023,gallego2024effect}.  

   \item There exist cases where there cannot exist a kinematic domain for topological reasons. 
   For example, suppose we wish to consider all solutions to the autoparallel equation. In spacetime this corresponds to all timelike, lightlike and spacelike geodesics simultaneously. In this case $U=TM$.
   Thus we would have to pick an initial velocity for each direction in order to construct the phase space $E$ upon which we define our Vlasov field.  Each fibre $E_p$ would then be topologically equivalent to the quotient set $T_pM/\sim$ where $\underline{v}\sim\lambda\underline{v}$ for $\lambda\neq 0$. This set is the real projective $(n{-}1)$--space $\mathbb{R}P^{n-1}$ in the case where $M$ is a $n$--dimensional spacetime manifold. However $\mathbb{R}P^{n-1}$ cannot be embedded into $T_pM$, for $n>1$, and hence no kinematic domain $E$ exists. By contrast our formalism remains valid in this case.
   
   \item Since kinematic domains consist only of future time pointing vectors, the existence of $E$ assumes that the system we are considering be time orientable. 
   This is the case for the timelike vectors, but fails when considering all vectors. 
   By contrast out approach works even if $U$ is not time orientable, such as when $U=TM$.

   \item Sometimes it is necessary to work in a particular kinematic domain for practical reasons, such as when performing numerical simulations of astrophysical plasmas \cite{warwick_moment_2023,ginzburg_simulation_2016}. 
   As stated above there are choices of kinematic domains, for example $\EH$ and $\ELab{t}$, and it may be necessary to transform the Vlasov equation between them. One can construct this transformation by considering the underlying ODEs, rescaling them, and then reconstructing the new corresponding Vlasov equation. The advantage of starting with the parameterisation free approach, is that it gives the formula for this transformation directly. This can be visualised in \cref{Fig Hyperboloid and Lab Bundles}.
   This is analogous to the advantages of working in coordinate-free notation. 
   If one is subsequently given a coordinate system, one can easily calculate the corresponding coordinate quantities from the coordinate-free quantities.
   Also, given two coordinate systems, the transformation of these coordinate-quantities is also derived from the coordinate free definitions. 
   Likewise the formula for passing from one kinematic domain to another falls out of our kinematic domain free definition of the Vlasov field.
   
   \item In the case of $\EH$ the kinematic domain is defined using the metric. Our approach works even when the manifold does not possess a metric. For example for the autoparallel equations, one can construct the Vlasov equations with just a connection and no metric. Thus it is compatible with pre-metric formalisms of dynamics (see \cite{hehl_foundations_2003,jancewicz_premetric_2008,itin_premetric_2018}). 
   For electrodynamics, without a metric, one would need to consider a force tensor, analogous to the electromagnetic field, but which mapped vectors to vectors.
    
   \item This formalism can be generalised to Finsler spacetimes \cite{javaloyes_definition_2020}. 
   Furthermore, our formalism is not intrinsically dependant on any geometric objects beyond a base manifold.
   By casting objects from Finsler geometry based kinematic theories (e.g. \cite{hohmann_kinetic_2020}), the dependence of these objects upon the Finsler metric can be better highlighted in a way analogous to pre-metric electromagnetism.
    
    \item The clock hypothesis is defined using the metric. Particles which decay must have a notions of time. However for stable particles, one can argue that the description of the particles motion should not be defined in terms of proper time which does not effect the particles. Our approach may be useful for treatments which do not impose the clock hypothesis, e.g. Mashoon electrodynamics \cite{mashhoon_complementarity_1988,mashhoon_hypothesis_1990,mashhoon_hypothesis_2004,mashhoon_nonlocal_2004}.

    \item Finally, there is a philosophical argument. There is vague distinction between the  kinematics and the dynamics of a system. In the case of particle dynamics, the kinematics simply state that we are interested in curves that satisfy some unspecified ODE, whereas the dynamics prescribe the ODE. It is the dynamics that require the connection and, maybe a metric. By contrast in standard Vlasov on $\EH$, the kinematics are defined on the kinematic domain that requires a metric. Thus the metric is introduced at the kinematic stage rather than the dynamics stage.

\end{itemize}


\subsection{The Vlasov Bivector and the Transport equations}

\begin{figure}[tb]
\centering
\includegraphics[scale=0.65]{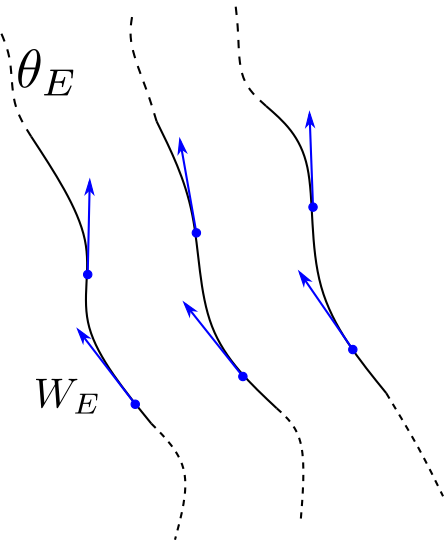}
\caption{\label{Fig Transport Equations} Diagram of the transport equations on a $(2n-1)$--dimensional kinetic domain $E$ using form submanifolds.
The details of submanifolds are described in \cite{gratus_pictorial_2017}.
The form manifolds of $\theta_E\in \Gamma\Lambda^{2n-2} E$ (black lines) do not terminate ($d\theta=0$) and are tangent to the vector field $W_E\in\Gamma TE$, represented by blue arrows ($i_{W_E}\theta_E=0$).}
\end{figure}

The primary object of interest in this paper is the Vlasov bivector, $\Psi$.
The Vlasov bivector is constructed on the conic bundle $U$ as oppose to a given kinematic domain.
Informally we may consider these to be the generalisation of a Vlasov field $W_E$ on $E$ to a geometric object on $U$.
Although we use the language of bivectors here, one can use the equivalent language of foliations to describe a Vlasov bivector.
Each $\Psi$ is constructed in such a way that the 2--dimensional leaves, 
\cref{fig:intro Leaves} and \cref{Geometric Bivector} that compose their foliation intersect any given kinematic domain $E$ to produce 1--dimensional curves, \cref{Magic Formula Diagram}.
These curves are exactly the integral curves of a Vlasov field associated with $\Psi$ which is tangent to $E$.

\begin{figure}[tb]
\centering
\includegraphics[scale=0.6]{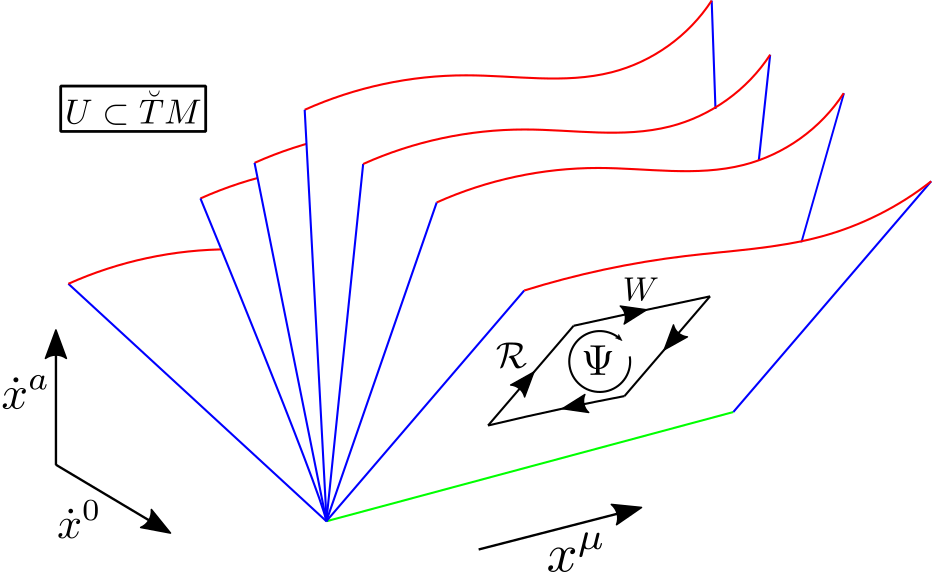}
\caption{\label{Geometric Bivector} 
Sketch of an integrable Vlasov bivector $\Psi$. Here, a possible form for $\Psi$ is $\Psi=\mathcal{R}\wedge W$. 
Notice that since $\Psi$ is integrable, the bivectors `knit-together' to from the leaves of a foliation. 
The ambient space is $U\subset \TMSlit$. The green line is to indicate the absence of 0-vectors in our space.
The density of the leaves correspondence to the velocity density of our one particle distributions function with higher density towards the middle and lower density towards the sides.
Viewing this diagram as the particle density form $\theta$, then the observation that the leaves are tangent to $\Psi$ and that they have no boundary, is equivalent to the transport equations given in \cref{PF Transport Equations}.}
\end{figure}

Vlasov bivectors also have the added benefit that they account for all projectively related Vlasov fields.
Consequently, the Vlasov bivector accounts for all parameterisations in tandem without explicit reference to any of them.
Furthermore, the equation for transforming Vlasov fields between kinematic domains can be easily derived from the Vlasov bivector, hence the choice to describe our formalism in terms of bivectors instead of foliations.
Vlasov bivectors are the main topic of discussion in \cref{Sec Vlasov Bivectors}.

In general it is not informative to attempt to extend the particle density function $f_E$ to $U$. Instead, we first replace $f_E$ with the particle density 6-form $\theta_E$. This is depicted in \cref{Fig Transport Equations}. This 6-form has two conditions, equivalent to \cref{Intr Vlasov Eqn on E} which together we call the {\em transport equations}. The first is that it is closed, $d\theta_E=0$ which corresponds to the lines not terminating (i.e. particles being created to lost). The second is that it is tangent to the Vlasov field $i_{W_E}\theta_E=0$. 
Furthermore, the transport equations apply to a more general set of theories than simply the Vlasov equation.
The transport equations are discussed in \cref{Subsec Transport Equations}.

The geometric notion of the transport equations can be more easily translated into our formalism.
In our case, the particle density from on $E$ $\theta_E$ is replaced with the particle density form $\theta\in\Gamma\Lambda^{2n-2}U$. This $(2n{-}2)$--form on a $2n$--dimensional manifold is also depicted in figure \cref{fig:intro Leaves}, using the idea of form-submanifolds described in \cite{gratus_pictorial_2017}. The closure of $\theta$ corresponds to the fact that the form-submanifolds do not have boundaries in $U$, they are also tangent to $W$. The regions of higher particle density correspond to the leaves being denser.
In \cref{Subsect Transport Equations on U} the transport equations are generalised to $U$ in terms of the Vlasov bivector and the particle density form on $U$.
These transport equations on $U$ can be reduced to transport equations on any given $E$ provided suitable conditions are satisfied.

Having established the Vlasov bivector, $W$, and the particle density form $\theta_E$, we show how to relate them to the corresponding Vlasov vector $W_E$ and density form $\theta_E$. Of course this is only possible if we are dealing with a time orientable system. We say that the conic bundle is time orientable if we can write $U$ in terms of a disjoint union of future pointing and past pointing vectors, discussed in section \cref{Subsec Specific Notation}. When $U$ is time orientable, we show how to translate between $W$ and $W_E$, and between $\theta$ and $\theta_E$.


\subsection{Signposting}

The remainder of section 1 is devoted to the notation and conventions used throughout this paper.
In section 2 we begin by introducing the kinematic domain and its properties in \cref{Subsect KD}: a generalisation of the mass shell and lab time bundles upon which we can perform plasma kinematics.
In this subsection we also introduce the Vlasov field $W$ on the conic bundle $U$ and explore its relationship with the Vlasov field $W_E$ on kinematic domains.
A useful tool for describing a kinematic domain is a \Fkin, these are introduced and discussed in \cref{Subsec Kinematic Indicators}.
Integral curves of Vlasov fields are discussed in \cref{Subsec ICs}.
We also discuss the conditions for which two different Vlasov fields correspond to the same trajectories.
In \cref{Sec Transforming Between Domains} we discuss how to transform Vlasov fields between kinematic domains in such a way that the trajectories of particles are preserved. 
In \cref{Subsec Transport Equations} we introduce the transport equations, a geometric method of interpreting the Vlasov equation on $E$.
We also discuss some advantages of this approach.

In \cref{Sec Vlasov Bivectors} we introduce the Vlasov bivector $\Psi$.
First, the necessary properties to define Vlasov bivectors are discussed in \cref{Subsect Bivectors} and \cref{Subsect Horizontal Bivectors}.
Vlasov bivectors themselves, their benefits and geometric interpretation are given in \cref{Subsect Vlasov Bivectors}. 
We also illustrate the correspondence between bivectors and foliations here.

Section 4 deals with the particle density form on $U$ and its applications.
The transport equations are generalised to the conic bundle using the Vlasov bivector in \cref{Subsect Transport Equations on U}.
We also discuss the geometric interpretation of these new transport equations here.
In \cref{Subsec PDF on U into E}, we explore the conditions necessary to define a particle density form on $U$ given a particle density from on a kinematic domain $E$ and vice versa.
We then apply the particle density on $U$ to define a current $(n{-}1)$--form on $M$ in \cref{Subsec Current Density}.
We also show that this current form is the same as the current form typically defined on $E$. In section \cref{subsect Stress energy} we define the stress-energy 3-forms. We see that unlike the case of the current, these depend on choice of kinematic domain $E$. 
We conclude in \cref{Sec Conclusion}.

\begin{figure}[tb]
\centering
\includegraphics[scale=0.55]{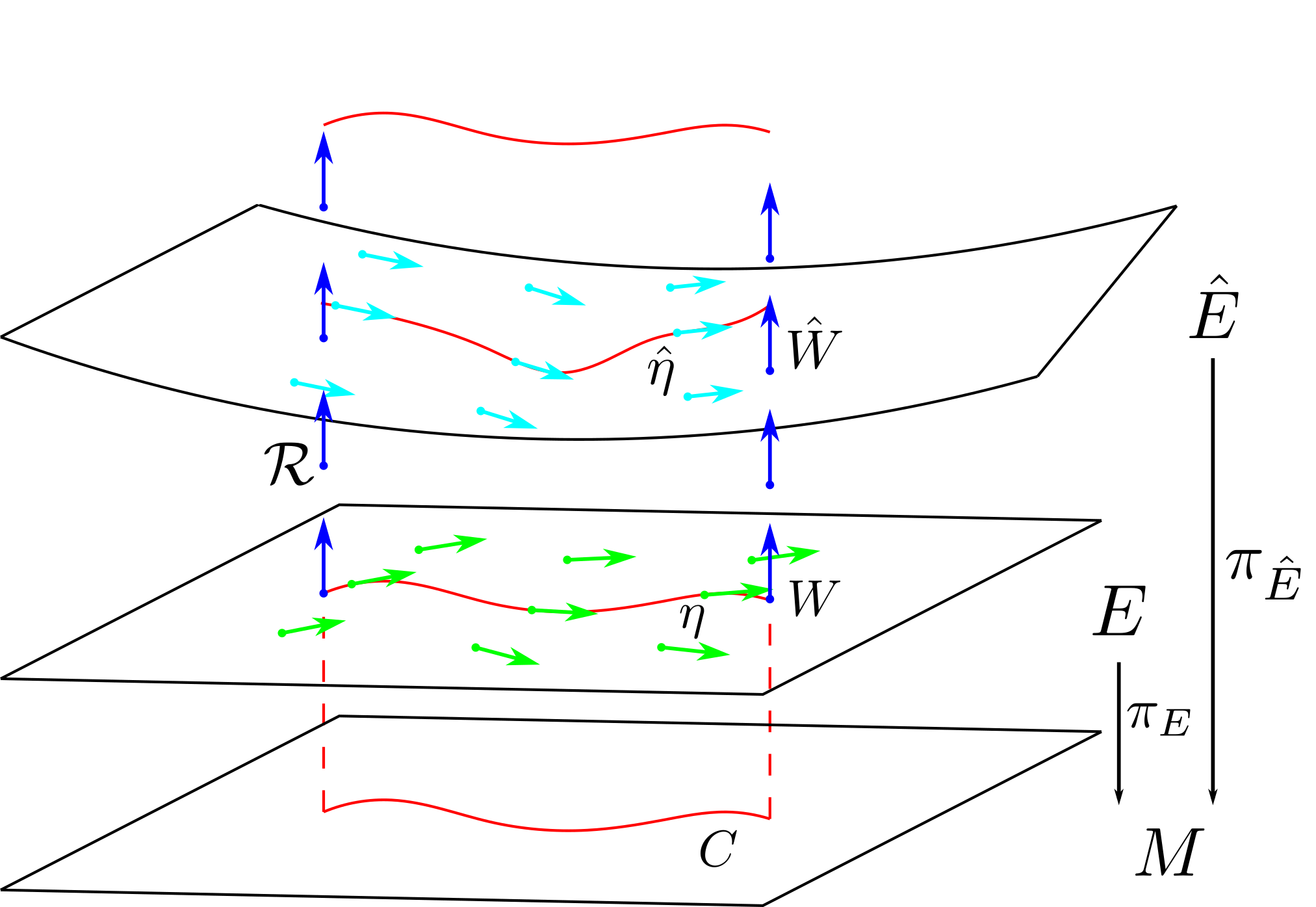}
\caption{\label{Magic Formula Diagram}
Given a two dimensional manifold $M$ and two different kinematic domains $E$ and $\hat{E}$, we can observe the relationship between two Vlasov fields $W$ and $\hat{W}$ related by \cref{Thm Magic Formula}.
The Vlasov field $W$ (resp. $\hat{W}$), depicted by green (light blue) arrows, are tangent to $E$ ($\hat{E}$) and generate integral curves $\eta$ ($\hat{\eta}$), depicted by red lines. Both integral curves project down into the same curve on $M$, denoted $C$, also a red line.
The radial vector field $\mathcal{R}$ is denoted by dark blue arrows.
This diagram can be identified with \cref{Fig Hyperboloid and Lab Bundles}.}
\end{figure}

\subsection{Notation and Conventions for general manifolds}
\begin{table*}
\begin{center}
\begin{tabular}{|c|c|c|c|}
\hline 
Object Type & 
    Set of Objects
& Generic & Specific\\
\hline
Manifolds & N.A. & $N,P,K$ & $M,E,U$\\
\hline
Scalar Fields & $\Gamma\Lambda^0N$ & $h,k$ & $f,F,g$\\
\hline
Vector Fields & $\Gamma TN$ & $X,Y,Z,V$ & $W$\\
\hline
$p$-forms & $\Gamma\Lambda^pN$ & $\alpha,\beta,\gamma,\omega$ & $\theta,\theta_E,\chi,\mathcal{J},\mathcal{J}_E$\\
\hline
Bivectors & $\Gamma\mathcal{B}^2(N)$ & $\Phi$ & $\Psi$\\
\hline
\end{tabular}
\end{center}
\caption{
Table of symbols use throughout.
The objects in the column marked specific are reserved for specific objects e.g. $W\in\Gamma TU$ is reserved for the Vlasov field, $\theta$ is used for the particle density form etc.}
\label{tab of symbols}
\end{table*}

Throughout this paper Greek indices will run from 0 to $n-1$ ($\mu,\nu=0,1,2,...,n-1$) and Latin indices will run from 1 to $n-1$ ($a,b=1,2,...,n-1$), unless specified otherwise.
We will use $M$ to denote a connected $n$-dimensional manifold and $TM$ to denote the tangent bundle over $M$ (all manifolds are assumed to be smooth).
Points in $M$ will generally be denote $p,q\in M$.
Vectors on $M$ which are points in $U$ will be denoted $\underline{u}\in U$.
The space of vector fields over an arbitrary smooth manifold $N$ of dimension $\ell$ is denoted $\Gamma TN$ and its elements are denoted $X,Y\in\Gamma TN$.
Maps between abstract manifolds are denoted $\Theta:N\rightarrow P$.
The space of scalar fields is denoted $\Gamma\Lambda^0N$ and we denote its elements $f,g,h,F,G,H\in\Gamma\Lambda^0N$ depending on context.
The space of $r$--forms is given by $\Gamma\Lambda^r N$ and we denote its elements by $\alpha,\beta\in\Gamma\Lambda^r N$. Although we consider both vectors at points and vector fields, we only consider $r$--form fields. Hence when referring to $r$--forms, the word field is implicit.
A field evaluated at a point will be denoted by $u\vert_{\underline{p}}$ and the action of a vector field on a scalar field $u\langle h\rangle$. A table of the generic as specific symbols is given in table \ref{tab of symbols}.

The internal contraction of an $r$--form $\alpha$ by a vector field $X$ is given by $i_X\alpha$ (or $\alpha{:} X$ for the special case when $\alpha$ is a 1--form), and the Lie derivative of $\alpha$ along $X$ is denoted by $L_X\alpha$.
Contraction by coordinate partial derivatives are denoted $i_\mu^{(x)}=i_{\partial_\mu^{(x)}}$ etc.
A similar convention is used for Lie derivatives along a coordinate vector field where we denote $L_{\partial_\mu^{(x)}}=L_\mu^{(x)}$ etc.

\begin{definition}[Pullback and Pushforward]\label{Def Pull and Push}
Given a map between manifolds $\Theta\colon N\rightarrow P$ , then the \textit{pushforward} of a vector at a point $\underline{v}\in T_p N$ is given by
\begin{equation}\label{Eqn Pushgorward}
\left( \Theta_\ast \underline{v} \right)\langle h \rangle = \underline{v}\langle h\circ \Theta\rangle.
\end{equation}
The \textit{pullback} of a scalar field $h\in\Gamma\Lambda^0P$ is given by
\begin{equation}
\Theta^\ast h= (h\circ \Theta).
\end{equation}
We may also pullback for forms of arbitrary degree $r$ according to the following rules:
\begin{align}
    \Theta^\ast (d h)=& d (\Theta^\ast h),\\
    \Theta^\ast(\alpha\wedge\beta)=& \Theta^\ast\alpha\wedge \Theta^\ast\beta,
\end{align}
for $h\in\Gamma\Lambda^0N,\; \alpha\in\Gamma\Lambda^1 N,\; \beta\in\Gamma\Lambda^{r-1}N$.
The pullback and pushforward satisfy the compatibility property:
\begin{equation}
\alpha\colon \Theta_\ast \underline{v} = \Theta^\ast\alpha \colon \underline{v},
\end{equation}
for any 1--form $\alpha\in\Gamma\Lambda^1N$ and point vector $\underline{v}\in T_pN$.
Observe that we pushforward vectors at points, while we pull back $r$--forms. 
\end{definition}

\begin{definition}[Tangential Vector Fields]\label{Def Tangent}
Given an embedding $\Theta\colon K\hookrightarrow N$,
a vector field $X\in\Gamma TN$ is \textit{tangent} to the submanifold $K$ if, there exists a vector field $Y\in\Gamma TK$ such that
\begin{equation}\label{Eqn Tangent}
\Theta_\ast(Y\vert_{p})= X\vert_{\Theta(p)}, \; \forall p\in K.
\end{equation}
Note that if $Y$ exists then it is unique.
In this case, we say that $Y\in\Gamma TK$ is induced by $X$.
Given an $r$--form $\alpha\in\Gamma\Lambda^r N$ then
\begin{align}
\Theta^\ast (i_X \alpha) 
=
i_{Y} (\Theta^\ast \alpha ).
\end{align}
\end{definition}

\begin{definition}[Scalar Lift]\label{Def Scalar Lift}
Given a scalar field $h\in\Gamma\Lambda^0N$, the scalar lift defines a \textit{scalar field} $\dot{h}\in\Gamma\Lambda^0TN$ given by
\begin{equation}
\dot{h}\vert_{\underline{u}}= \underline{u} \langle h \rangle.
\end{equation}
\end{definition}

Given a coordinate system $(x^0,\ldots,x^{n-1})$ for the patch $K\subset N$, then this induces a coordinate system $(\bar{x}^0,\ldots,\bar{x}^{n-1},\dot{x}^0,\ldots,\dot{x}^{n-1})$ for $\pi^{-1}(K)\subset TN$, where $\bar{x}^\mu=\pi^\ast x^\mu$ and $\dot{x}^\mu$ is the scalar lift. In this coordinate system an arbitrary vector field $X\in\Gamma TN$ can be written 
\begin{equation}\label{Eqn Vector Field in Coords}
\begin{split}
    X= 
\bar{X}^\mu \frac{\partial \;}{\partial \bar{x}^\mu}+ 
\hat{X}^\mu\frac{\partial \;}{\partial \dot{x}^\mu}
\quad\text{where}\\
\quad
\bar{X}^\mu = X\langle{\bar{x}^\mu}\rangle
\text{ and }
\hat X^\mu = X\langle{\dot{x}^\mu}\rangle
\end{split}
\end{equation}
Every vector field over $TN$ can be defined using \cref{Eqn Vector Field in Coords} and consequently can be defined entirely by its action on scalar fields of type $\pi^\ast h$ and $\dot{h}$ for all $h\in\Gamma\Lambda^0N$.

Throughout we will use the shorthand 
\begin{equation}
\partial_\mu^{(x)}= \frac{\partial\;}{\partial x^\mu},\quad 
\partial_\mu^{(\bar{x})}= \frac{\partial\;}{\partial \bar{x}^\mu},\quad
\partial_\mu^{(\dot{x})}= \frac{\partial\;}{\partial \dot{x}^\mu}.
\end{equation}
Note that $\partial_\mu^{(x)}\in\Gamma TM$ while $\partial_\mu^{(\bar{x})}\in\Gamma TTM$ and  $\partial_\mu^{(\dot{x})}\in\Gamma TTM$.
In induced coordinates, the scalar lift of a scalar field $h\in\Gamma\Lambda^0M$ is given by
\begin{equation}
\dot{h}= \dot{x}^\mu  \partial^{(x)}_\mu h .
\end{equation}


\subsection{Notation and Conventions specific for this article}\label{Subsec Specific Notation}

Let $M$ be a spacetime manifold and $\pi:TM\to M$, the projection from the tangent bundle to the base manifold. The slit tangent bundle $\TMSlit=TM\backslash\{0\}$ is the subset of $TM$ which excludes the zero vectors. 

\begin{definition}[Conic Sub-Bundle]\label{Def Conic Bundle}
A \textit{conic sub-bundle} is a subset of the slit tangent bundle $U\subset \TMSlit$, which satisfies the following properties:
\begin{enumerate}
\item $\pi(U)=M$
\item If $\underline{u}\in U$ then for any $\lambda\neq 0$ we have $\lambda\underline{u}\in U$.
\end{enumerate}
We also use $\pi$ to also denote the restriction $\pi:U\to M$
A formal discussion of conic bundles can be found in \cite{javaloyes_definition_2020}.
We additionally assume the properties necessary for $U$ to have a smooth sub-bundle structure.
\end{definition}

\begin{definition}[Causal Indicator]\label{Def Causal Indicator}
We say $U$ is {\it time orientable} if $U=U^+\cup U^-$ where $U^+\cap U^-=\varnothing$ and $\underline{u}\in U^+$ if and only if $-\underline{u}\in U^-$. We call $\underline{u}\in U^+$ {\it future pointing}.
The \textit{causal indicator} is the scalar field  $\sigma\in\Gamma\Lambda^0U$ given by
\begin{equation}
\sigma\vert_{\underline{u}}=
\begin{cases}
1, & \mbox{if }\underline{u}\in U^+
\\
-1, & \mbox{if } \underline{u}\in U^-
\end{cases}
\end{equation}
\end{definition}

Examples of structures described by conic bundles include the collection of timelike vectors over a spacetime manifold, the light cone, and the causal region of a time orientable manifold.
In the case where consider only the timelike conic bundle the corresponding conic bundle has an open sub-bundle structure.
In the case where we consider the causally connected regions of spacetime (the former case plus the light cone), the corresponding conic bundle is a sub-bundle with a boundary. However it is not a closed set since it excludes the zero vectors.

As stated in the introduction, an important example of conic bundle which is not time orientable is the entire $U=\TMSlit$, for any manifold of dimension greater than 1. This is because each fibre, $\breve{T_p}M$, is connected. 

As discussed in the introduction, when dealing with kinematic domains 
we restrict our attention to time orientable conic bundles. By contrast when dealing with the Vlasov bivector approach we do not have to make this restriction. Thus the bivector approach can be applied to the case when $U=\TMSlit$.

\begin{definition}[Radial Vector Field]
The \textit{radial vector field}, denoted by $\mathcal{R}\in \Gamma TU$, is the unique vector field that satisfies the following condition: for any $h\in\Gamma \Lambda^0M$, $\underline{u}\in U$ and $\pi\colon TM\rightarrow M$ we have
\begin{equation}
\mathcal{R}\langle \pi^\ast h\rangle =0, \quad \mbox{ and } \quad \mathcal{R}\langle \dot{h} \rangle \vert_{\underline{u}}= \dot{h}\vert_{\underline{u}}= \underline{u}\langle h \rangle .
\end{equation}
In local induced coordinates $\mathcal{R}$ is given by
\begin{equation}
\mathcal{R}= \dot{x}^\mu \partial^{(\dot{x})}_\mu.
\end{equation}
This can be seen by acting $\mathcal{R}$ on the induced coordinates: $\mathcal{R}\langle \bar{x}^\mu\rangle=0$, $\mathcal{R}\langle \dot{x}^\mu\rangle = \dot{x}^\mu$.
Within the literature, the radial vector field is known by many names, perhaps most common of which is the vertical vector field. 
We have opted to refer to it as the radial vector field as it exhibits many useful radial properties.
\end{definition}

\begin{definition}[Homogeneity]
A scalar field $G\in\Gamma\Lambda^0U$ is said to be \textit{homogeneous of degree} $k$ (also known as radially homogeneous and fibre-wise homogeneous)  if \begin{equation}
    G\vert_{\lambda{\underline{u}}}= \lambda^k G\vert_{\underline{u}}
    \quad\text{for all}\quad\lambda\in\mathbb{R}\setminus\{ 0 \}.
\end{equation}
\end{definition}
Note that by Euler's theorem of homogeneous functions, a function $G\in\Gamma \Lambda^0U$ is (fibre-wise) $k$--homogeneous if and only if it satisfies
\begin{equation}\label{Eqn Euler Thm}
\mathcal{R}\langle G\rangle =kG.
\end{equation}

\section{Vlasov Systems and Kinematic Domains}
\label{Sect Vlasov E}

\subsection{Kinematic Domains and the Vlasov Picture}\label{Subsect KD}

In order to discuss the Vlasov field we must first establish the space over which it is defined. 
Conventionally, when we perform kinematics with the Vlasov equation, this is done on a choice of 7-dimensional submanifold $E\subset U$.
We call such submanifolds kinematic domains.
For this section we assume that $U$ is time orientable in the sense of \cref{Def Causal Indicator}, and that $E\subset U^+$, that is all points in $E$ are future pointing. As stated in the introduction, this constraint is not needed for the generalisation to Vlasov bivectors, discussed below.
We assume that each spacetime manifold $M$ (and hence $TM$) is connected.

\begin{definition}[Kinematic Domain]\label{Def Kinematic Domain}
A \textit{Kinematic domain} is a $(2n{-}1)$--dimensional submanifold $E\subset U^+$ which satisfies the following properties:
\begin{enumerate}
\item $\pi_E:E\rightarrow M$ is surjective,
\item $E$ is connected,
\item For any $\underline{u}\in U$ there exists a unique $\underline{v}\in E$ and $\lambda\neq 0$ such that $\underline{u}=\lambda \underline{v}$.
\end{enumerate}
\end{definition}

For example, the unit mass shell, or as we refer to it throughout the remainder of this paper, the upper unit hyperboloid, given in \cref{Eqn Intro EH}, is a kinematic domain. We can write this as
\begin{equation}\label{Unit Hyperboloid Bundle}
\EH=\{ \underline{u}\in U:  \FkinUH=1 \text{ and } \sigma|_{\underline{u}}=1 \},
\end{equation}
where
\begin{equation}\label{Eqn Fkin UH}
\FkinUH\colon U\rightarrow \mathbb{R};\; \underline{u}\mapsto -g_{\pi(\underline{u})}(\underline{u},\underline{u}),
\end{equation}
Another example of a kinematic domain would be the lab time bundle, \cref{Eqn Intro ELab}. 
An example from Finsler geometry\footnote{Recall, examples from Finsler geometry may be ignored for readers unfamiliar with this geometry. For this reason we do not give definitions of standard Finsler objects.} is defined in \cite{hohmann_non-metric_2016}, where the Finsler spacetime $(M,L,F)$ admits an observer bundle
\begin{equation}
    \mathcal{O}= \bigcup_{p\in M} \mathcal{S}_p,
 \end{equation}
where
\begin{equation}
    \mathcal{O}_p= \left\{ 
    \begin{array}{c}
       \underline{u}\in T_pM \; \colon \; L\vert_{\underline{u}}=\pm1, \\
       g^{(L)}_{\mu\nu}\vert_{\underline{u}} 
       \mbox{ has signature } (L,-L,-L,-L) 
    \end{array}
    \right\},
\end{equation}
and $\mathcal{S}_p$ is a non-empty closed connected component of $\mathcal{O}_p$.
Here $g_{\mu\nu}^{(L)}$ is the metric induced by the fundamental function $L$.

We refer to two types of Vlasov fields throughout this paper: Vlasov fields on kinematic domains (denoted $W_E\in\Gamma TE$ for $E\subset U$), and Vlasov fields in $U$ (denoted $W\in\Gamma TU$).

\begin{definition}[Vlasov field on $E$]
Given a kinematic domain $E$, a \textit{Vlasov field on} $E$ then $W_E\in\Gamma TE$, is a Vlasov vector field, if it satisfies the horizontal condition
\begin{equation}
\pi_\ast(W_E\vert_{\underline{v}})=\underline{v}, \mbox{  or equivalently,  }  W_E\langle \pi^\ast h \rangle =\dot{h},
\end{equation}
for any $\underline{v}\in E$ and $h\in\Gamma\Lambda^0E$.
\end{definition}
Given a particle distribution function $f_E\in\Gamma\Lambda^0E$, the Vlasov field $W_E\in\Gamma TE$ satisfies \cref{Intr Vlasov Eqn on E}.
The flow of the Vlasov field defines a set of integral curves which trace out the paths of particles through phase space (see \cref{Subsec ICs}).
The projections of these integral curves into the base space $M$ are the trajectories of the particles described by $f_E$.

An example of the Vlasov field on $\EH$ is the Lorentz force Vlasov equation.
In local coordinates $(t,x^a,v^a)$ on $\EH$, define $v^0$ by solving the constraint $g_{\mu\nu}v^\mu v^\nu=-1$ and choosing the root, such that $v^\mu\partial_\mu\in U^+$. 
In this coordinate system, the Lorentz force Vlasov field is given by
\begin{equation}
W_{\EH}= v^0 \partial_t+ v^a\partial_{a}^{(x)}- \left(\Gamma^{\mu}_{\nu\rho}v^{\nu}v^{\rho}
-\frac{q}{m}g^{\mu\nu}\mathcal{F}_{\nu\rho}v^\rho\right) \partial_{a}^{(v)}.
\end{equation}
Here $\mathcal{F}\in\Gamma\Lambda^2M$ is the Faraday 2--form  which satisfies the Maxwell equations.

\begin{definition}[Vlasov Field on $U$]\label{Def Vlasov Field on U}
Denoted $W\in \Gamma TU$, a \textit{Vlasov field on} $U$ is a vector field with the following defining properties:
\begin{enumerate}
\item $W$ is horizontal,
\begin{equation}\label{W horizontal}
\pi_\ast (W\vert_{\underline{u}})=\underline{u}, \mbox{  or equivalently,  }  W\langle \pi^\ast h \rangle =\dot{h} ,
\end{equation}
for any point vector $\underline{u}\in U$ and $h\in\Gamma\Lambda^0M$;
\item $W$ is radially quadratic, 
\begin{equation}\label{Eqn Radially Quadratic}
W\langle \dot{h}\rangle \vert_{\lambda \underline{u}} = \lambda^2 W\langle \dot{h}\rangle \vert_{\underline{u}},
\end{equation}
for any $h\in\Gamma\Lambda^0M$, $\underline{u}\in U$, and $\lambda\in\mathbb{R}^+$.
\end{enumerate}
From \cref{Eqn Radially Quadratic} and \cref{Eqn Euler Thm} the radially quadratic property can equivalently be stated as 
\begin{equation}\label{Eqn VQ Vlasov}
\mathcal{R}\langle W\langle \dot{h} \rangle\rangle =2W\langle \dot{h}\rangle.
\end{equation}
\end{definition}

\begin{lemma}\label{Lem VQ iff Com}
Let $W\in\Gamma TU$ be horizontal. Then $W$ is radially quadratic and hence a Vlasov field if and only if 
\begin{equation}\label{Involution}
[\mathcal{R},W]=W.
\end{equation}
\end{lemma}

\begin{proof}
Let $W$ be radially quadratic then for any $f\in\Gamma\Lambda^0M$, then $W\langle \dot{f} \rangle$ is a 2--homogeneous scalar field.
It follows that $\mathcal{R}\langle W\langle \dot{f}\rangle \rangle= 2W\langle \dot{f}\rangle$ by \cref{Eqn VQ Vlasov}.
Hence $[\mathcal{R},W]\langle \dot{f}\rangle= \mathcal{R}\langle W\langle \dot{f}\rangle \rangle-W\langle \mathcal
R\langle\dot{f}\rangle\rangle= W\langle \dot{f}\rangle$.
From \cref{W horizontal}, we also have 
\begin{equation*}\label{Eqn Phi Vertical Part}
[\mathcal{R},W]\langle \pi^\ast f\rangle =\mathcal{R}\langle \dot{f} \rangle = \dot{f} =W\langle \pi^\ast f\rangle,
\end{equation*}
and hence $[\mathcal{R},W]=W$.

Suppose that $[\mathcal{R},W]=W$, then for any $f\in\Gamma\Lambda^0M$ we have
\begin{equation}
W\langle \dot{f} \rangle= [\mathcal{R},W]\langle \dot{f}\rangle = \mathcal{R}\langle W\langle \dot{f}\rangle\rangle - W\langle \dot{f}\rangle.
\end{equation}
Rearranging the above gives $\mathcal{R}\langle W\langle \dot{f}\rangle \rangle = 2W\langle \dot{f}\rangle$.
Hence $W$ is radially quadratic.
\end{proof}

In figure \cref{Fig Hyperboloid and Lab Bundles} we see two Vlasov fields on $U$, that represent the same trajectories on $M$, just with different parameterisations. 
One is tangent to the unit hyperboloid $E_H$ while the other is tangent to a lab time bundle. In section \cref{Sec Transforming Between Domains} we see the formula for transforming from one Vlasov field to another in such a way that the particle trajectories are unaffected.

The Vlasov field on $U$ can be written in local coordinates as
\begin{equation}\label{Vlasov Coords}
W= \dot{x}^\mu \partial^{(x)}_\mu  +\varphi^\mu \partial_\mu^{(\dot{x})}.
\end{equation}
where $\varphi^\mu$ are 2--homogeneous scalar fields $\varphi^\mu\vert_{\lambda\underline{u}}=\lambda^2\varphi^\mu\vert_{\underline{u}}$.

The Vlasov field $W$ on $U$ can be reduced to a Vlasov field on $E$ provided it is tangent to it.
\begin{definition}[Extension of a Vlasov Field on $E$]\label{Def Extension of WE}
Given a kinematic domain $E$ with inclusion map $\Sigma_E:E\hookrightarrow U$ and a Vlasov field $W_E\in\Gamma TE$, we call $W\in\Gamma TU$ the extension of $W_E$ if $W$ is tangent to $E$ and $W_E$ is induced by $W$, as in \cref{Eqn Tangent},
\begin{equation}\label{Eqn Vlasov Pushforward}
\Sigma_{E}{}_\ast \left( W_{E}\vert_{\underline{v}} \right)= W \vert_{\Sigma_E(\underline{v})}.
\end{equation}

\end{definition}

\begin{lemma}\label{Lem Homogeneity of W<G>}
Let $W$ be a Vlasov field and $G$ be any $k$--homogeneous scalar field.
Then,
\begin{equation*}
W\vert_{\lambda\underline{u}}\langle G\rangle
=
\lambda^{k+1} W\vert_{\underline{u}}\langle G\rangle .  
\end{equation*}
\end{lemma}

\begin{proof}
Let $G$ be a $k$--homogeneous function.
First note that
\begin{equation*}
    \partial_\mu^{(\dot{x})}\left( G\vert_{\lambda\underline{u}} \right)= \lambda \left( \partial_\mu^{(\dot{x})} G \right)\big{\vert}_{\lambda\underline{u}},
\end{equation*}
by the chain rule.
Hence we have
\begin{equation*}
    \begin{split}
        \left( \partial_\mu^{(\dot{x})} G \right)\big{\vert}_{\lambda\underline{u}}
    &= \lambda^{-1} \partial_\mu^{(\dot{x})}\left( G\vert_{\lambda\underline{u}} \right)\\ 
    {}&= \lambda^{-1} \partial^{(\dot{x})}_\mu \left( \lambda^k G\vert_{\underline{u}} \right)\\
    {}&= \lambda^{k-1} \left( \partial_\mu^{(\dot{x})} G \right)\big{\vert}_{\underline{u}}.
    \end{split}
\end{equation*}
We also have
\begin{equation*}
    \left(\partial_\mu^{(x)}G  \right) \big{\vert}_{\lambda\underline{u}}=  
 \partial_\mu^{(x)}\left(\lambda^k G\vert_{\underline{u}}  \right) =\lambda^k \left(\partial_\mu^{(x)}G  \right) \big{\vert}_{\underline{u}}.
\end{equation*}
By expanding the Vlasov field in coordinates we get
\begin{equation*}
\begin{split}
    W\vert_{\lambda\underline{u}}\langle G\rangle & = \left( \dot{x}^\mu\partial_\mu^{(x)}G+ \varphi^\mu\partial_\mu^{(\dot{x})}G  \right)\big{\vert}_{\lambda \underline{u}}\\
    {}&= \lambda\underline{u} \lambda^k \left(\partial_\mu^{(x)}G  \right) \big{\vert}_{\underline{u}}\\
    {}& \quad +\lambda^2 \varphi^\mu\vert_{\underline{u}} \lambda^{k-1} \left( \partial_\mu^{(\dot{x})} G \right)\big{\vert}_{\underline{u}} \\
    {}&= \lambda^{k+1} \left( \dot{x}^\mu\partial_\mu^{(x)}G+ \varphi^\mu\partial_\mu^{(\dot{x})}G  \right)\big{\vert}_{\underline{u}}\\
    {}&= \lambda^{k+1} W\vert_{\underline{u}} \langle G\rangle,
\end{split}
\end{equation*}
hence the result.
\end{proof}

It is worth noting at this point that what we call a Vlasov fields on $U$ are referred to as sprays and $W_E$ are referred to as semi-sprays in the literature (see \cite{shen_differential_2001} or \cite{lang_vector_1999} for an overview of the theory of sprays and semi-sprays).
The correspondence between these objects is explored in \cref{Spray Section}.

We now give the example of the Lorentz force equation and the corresponding Vlasov fields adapted to the unit hyperbolid $\EH$ and the lab time $\ELab{t}$. Here $U$ is the conic bundle of timelike vectors.
The Vlasov field on $U$ adapted to $\EH$ is given by
\begin{equation}\label{Eqn LF Vlasov}
W=\dot{x}^\mu \partial^{(x)}_\mu
+ \left(\frac{q}{m} \sigma \sqrt{ \FkinUH } g^{\mu\nu}\mathcal{F}_{\nu\rho}\dot{x}^\rho- \Gamma^\mu_{\nu\rho}\dot{x}^\nu\dot{x}^\rho
 \right) \partial^{(\dot{x})}_\mu,
\end{equation}
where $\FkinUH$ is as defined in \cref{Eqn Fkin UH}.
The inclusion of the factor $\sigma \sqrt{\FkinUH}$ is to ensure the $\partial^{(\dot{x})}_\mu$ term is 2-homogeneous.
The parameterisation associated with the trajectories of this Vlasov field is proper time $\tau$. See \cref{Subsec ICs} for an overview of the  parameterisation of the integral curves of a Vlasov field.
The Vlasov field describing particles subject to the Lorentz force in the lab time bundle $W_{\ELab{t}}$ with associated lab time $t$ is given by
\begin{align}
\begin{split}
    \hspace{-0.6em}
    \hat{W} &= \dot{x}^\mu \partial^{(x)}_\mu
    + \bigg{(} \frac{q}{m} \sigma \sqrt{ \FkinUH } g^{\mu\nu}\mathcal{F}_{\nu\rho}\dot{x}^\rho-\Gamma^\mu_{\nu\rho}\dot{x}^\nu\dot{x}^\rho
    \bigg{)} \partial^{(\dot{x})}_\mu\\
    {}&- \bigg{(}  \frac{q}{m} \sigma \sqrt{\FkinUH} g^{\lambda\nu}\mathcal{F}_{\nu\rho}\dot{x}^\rho \frac{\partial t}{\partial x^\lambda} 
    - \Gamma^\lambda_{\nu\rho}\dot{x}^\nu\dot{x}^\rho \frac{\partial t}{\partial x^\lambda}\\
    {}& +\dot{x}^\nu\dot{x}^\rho \frac{\partial t}{\partial x^\nu\partial x^\rho}   \bigg{)} \frac{\dot{x}^\mu}{\dot{t}}  \partial_\mu^{(\dot{x})}.\label{Eqn LT Vlasov}
\end{split}
\end{align}
When using a lab time bundle it is often convenient to choose a coordinate system $(t,x^1,x^2,x^3)$ adapted to the lab time $t$. In this coordinate system \cref{Eqn LT Vlasov} becomes
\begin{align}
\begin{split}
    W_{\ELab{t}} &= \dot{x}^\mu \partial^{(x)}_\mu
    + \bigg{(} \frac{q}{m} \sigma \sqrt{ \FkinUH } g^{\mu\nu}\mathcal{F}_{\nu\rho}\dot{x}^\rho
    -\Gamma^\mu_{\nu\rho}\dot{x}^\nu\dot{x}^\rho
    \bigg{)} \partial^{(\dot{x})}_\mu\\
    {}&- \left(  \frac{q}{m} \sigma \sqrt{\FkinUH} g^{0\nu}\mathcal{F}_{\nu\rho}\dot{x}^\rho- \Gamma^0_{\nu\rho}\dot{x}^\nu\dot{x}^\rho   \right) \frac{\dot{x}^\mu}{\dot{t}}\partial_\mu^{(\dot{x})},\label{Eqn LT Vlasov Coord}
\end{split}
\end{align}
where $x^0=t$.
We see below in subsection \ref{subsec Vlasov LF}, that we can demonstrate the transformation between \cref{Eqn LF Vlasov} and \cref{Eqn LT Vlasov} using the underlying 2nd order ODE. However it is much easier to calculate these transformation after we have defined the kinematic indicator.


\subsection{Kinematic Indicators}\label{Subsec Kinematic Indicators}

A kinematic domain can be defined in terms of a scalar field on $U$, we call a kinematic indicator.

\begin{definition}[Kinematic Indicator]\label{Def Kinematic Indicator}
A \textit{\Fkin} for $E$, is a non-vanishing scalar field $F\in\Gamma\Lambda^0 U$ with non zero integer degree of homogeneity $k$, $F\vert_{\lambda \underline{v}}=\lambda^k F\vert_{\underline{v}}$, such that $F$ is positively valued on $U^+$, and $E$ is given by 
\begin{equation}\label{Eqn E hat}
E=\{ \underline{u}\in U \colon F\vert_{\underline{u}}=a \; \& \; \sigma\vert_{\underline{u}}=1\}
\end{equation}
for some positive number $a$.
Note that if $k$ is odd the inclusion of the $\sigma\vert_{\underline{u}}=1$ condition is unnecessary.
\end{definition}
For example $\FkinUH$, given in \cref{Eqn Fkin UH} is a kinematic indicator for $\EH$, the unit hyperboloid. A second example $\dot{t}$, is a kinematic indicator for $\ELab{t}$, the lab time bundle.

\begin{lemma}\label{Lem F from E}
Given an arbitrary kinematic domain $E$, there exists a unique 1--homogeneous \Fkin $F\in\Gamma\Lambda^0 U$ such that $E=\{ \underline{u}\in U : F\vert_{\underline{u}}=1 \}.$

For each $\underline{u}\in U$ this is given by 
\begin{equation}
    F\vert_{\underline{u}}=\lambda \quad \mbox{where $\;\underline{u}=\lambda\underline{v}\;$  for a unique $\underline{v}\in E$.}
    \label{Eqn F11}
\end{equation}
\end{lemma}

\begin{proof}
For each $\underline{u}\in U$, there exists a unique $\lambda\in\mathbb{R}\setminus\{0\}$ and $\underline{v}\in E$ such that $\underline{u}=\lambda\underline{v}$ by the conic properties of $U$ (\cref{Def Conic Bundle}). 
We define then $F\vert_{\underline{u}}=\lambda$.
Since each component of $E$ is connected $F$ is smooth and hence $F\in\Gamma\Lambda^0 U$.
To see that $F$ is unique, suppose there is another scalar $F'$ which is 1--homogeneous and is such that such that $E=\{ \underline{u}\in U : F\vert_{\underline{u}}=1 \}$ and pick any $\underline{u}\in U$. 
Then $F'\vert_{\underline{u}}=F'\vert_{\lambda\underline{v}}=\lambda=F\vert_{\underline{u}}$.
\end{proof}

Given a \Fkin $F$ for a kinematic domain $E$ as defined in \cref{Lem F from E}, we may define another scalar for some $k\in\mathbb{Z}$, $k\ne0$ and $a\in\mathbb{R}^+$ 
\begin{equation}\label{Eqn F hat}
\hat{F}= aF^k.
\end{equation}
Not only is $\hat{F}$ a $k$--homogeneous \Fkin for $E$ such that
it is also unique for the chosen $k$ and $a$.

\begin{lemma}
Let $F$ be a \Fkin for $E$ as given by \cref{Lem F from E} and let $\hat{F}$ be given by \cref{Eqn F hat}. 
$\hat{F}$ is a \Fkin for $E$.
\end{lemma}

\begin{proof}
Let $\hat{E}$ be a contour of $\hat{F}$ as given by \cref{Eqn E hat}.
For and $\underline{v}\in E$ notice that $F^k\vert_{\underline{v}}=1$ so that
$\hat{F}\vert_{\underline{v}}=aF^k\vert_{\underline{v}}=a$.
Hence $\underline{v}\in E$ implies $\underline{v}\in\hat{E}$.
The converse can be proved similarly to show $\underline{v}\in\hat{E}$ implies $\underline{v}\in E$.
Hence $\hat{F}$ is a \Fkin for $E$.
\end{proof}

\begin{lemma}
Let $a\in\mathbb{R}^+$ and $k\in\mathbb{N}$.
If $\hat{F}$ is a $k$--homogeneous \Fkin for $E$ such that \cref{Eqn E hat} is satisfied, then $\hat{F}$ is uniquely given by \cref{Eqn F hat}.
\end{lemma}

\begin{proof}
$\hat{F}$ satisfies the following properties:
\begin{align*}
\hat{F}\vert_{\underline{v}}=a,& \quad \forall \underline{v}\in E\\
\hat{F}\vert_{\lambda\underline{u}}= \lambda^k \hat{F}\vert_{\underline{u}},& \quad \forall \underline{u}\in U \; \& \; \lambda\neq 0,
\end{align*}
These two functions can be related by
\begin{equation*}
    \hat{F}\vert_{\underline{v}}= a =a F\vert_{\underline{v}}, \quad \forall \underline{v}\in E
\end{equation*}
\begin{align}
\frac{(\sigma\hat{F})\vert_{\lambda\underline{u}}}{\hat{F}\vert_{\underline{u}}}= \sigma\vert_{\lambda\underline{u}}\lambda^k=\sigma\vert_{\lambda\underline{u}}\left( \frac{F\vert_{\lambda\underline{u}}}{F\vert_{\underline{u}}}  \right)^k ,&{}\label{Eqn Fhat and F}
\end{align}
for all $\underline{u}\in U$ and $\lambda\neq 0$.

By the conic property, for any $\underline{u}\in U$ there exists a unique $\underline{v}\in E$ and $\lambda \neq0$ such that $\underline{u}=\lambda\underline{v}$.
By plugging such a $\underline{v}$ into \cref{Eqn Fhat and F} we get
\begin{equation*}
\frac{\hat{F}\vert_{\lambda\underline{v}}}{a}= \left(  \frac{F\vert_{\lambda\underline{v}}}{1}  \right)^k \implies \hat{F}\vert_{\underline{u}}= aF^k\vert_{\underline{u}}.
\end{equation*}
Hence $\hat{F}=aF^k$.
\end{proof}

In the case of the upper unit hyperboloid, the \Fkin is 2--homogeneous.
The equivalent 1--homogeneous \Fkin is given by
\begin{equation}
\hat{F}_H\vert_{\underline{u}}= \sigma\vert_{\underline{u}} \sqrt{\FkinUH\vert_{\underline{u}}},
\end{equation}
where $\FkinUH$ is as given by \cref{Eqn Fkin UH} and $\sigma$ is given by \cref{Def Causal Indicator}.

At this point we may note some similarities with other foliations of the tangent bundle based on Finsler metrics \cite{bejancu_finsler_2006}.
The function $F$ which defines $E$ can be a Finsler metric, however it is not a necessary condition.
The only requirement of $F$ is that it be homogeneous of some degree.

\begin{definition}[Compatible Vlasov Field]
Given a \Fkin $F$, a Vlasov field $W\in\Gamma TU$ is said to be compatible with $F$ if it satisfies $W\langle F\rangle=0$.
The restriction of $W$ to a compatible $E$ is then denoted $W_E\in\Gamma TE$ as given by \cref{Eqn Vlasov Pushforward}.
\end{definition}

Note that by \cref{Lem WE Tangent}, $W$ is tangent to a kinematic domain $E$ with \Fkin $F$ if and only if $W$ is compatible with $F$.

\begin{lemma}\label{Lem WE to W}
Given a kinematic domain $E$ and a Vlasov field $W_E\in\Gamma TE$, there exists a unique Vlasov field $W\in\Gamma TU$ $W$ which is an extension of $W_E$, or equivalently $W_E$ is induced by $W$.
\end{lemma}

\begin{proof}
First observe that for any $f\in\Gamma\Lambda^0M$ and $\underline{v}\in E$,
\begin{equation*}
\left( \dot{f}\circ\Sigma_E \right) \vert_{\underline{v}}= \dot{f}\vert_{\Sigma_E(\underline{v})}=\dot{f}\vert_{\underline{v}}.
\end{equation*}
It follows that $\dot{f}\circ\Sigma_E= \dot{f}$.

For all $\underline{u}\in U$ and $f\in\Gamma\Lambda^0M$ define $W\in\Gamma TU$ by
\begin{equation*}\label{Eqn W from WE I}
    \pi^\ast W\vert_{\underline{u}}
    = \underline{u} \;\mbox{ and }\;
    W\vert_{\underline{u}}\langle \dot{f}\rangle 
    = \lambda^2 W_E\vert_{\underline{v}}\langle \dot{f}\rangle,
\end{equation*}
where $\underline{u}=\lambda\underline{v}$ by the conic property for $\underline{v}\in E$ and $\lambda\neq 0$.
By construction $W$ is radially quadratic and horizontal so it remains to show that it is unique.

Let $W$ and $\hat{W}$ both induce $W_E$ as defined above and set $X=\hat{W}-W$.
Observe that $\pi_\ast X=0$ and
\begin{equation*}
X\vert_{\underline{u}}\langle \dot{f}\rangle = \lambda^2(\hat{W}_E-W_E)\vert_{\underline{v}}\langle \dot{f}\rangle =0.
\end{equation*}
Hence it follows that $X=0$ and $W$ is unique.

To see that $W$ is compatible with $E$ (and hence tangent to it by \cref{Lem WE Tangent}) consider the following.
Without loss of generality we may assume $F$ is 1--homogeneous. 
By \cref{Lem Homogeneity of W<G>} we have $W\vert_{\underline{u}}\langle F\rangle= \lambda^2 W\vert_{\underline{v}}\langle F\rangle$.
It follows that
\begin{equation*}
    \begin{split}
        W\vert_{\underline{u}}\langle F\rangle &= \lambda^2 W\vert_{\underline{v}} \langle F\rangle\\
        {}&= \lambda^2 (\Sigma_E{}_\ast W_E\vert_{\underline{v}}) \langle F\rangle\\
        {}&= \lambda^2 W_E\vert_{\underline{v}} \langle F\circ \Sigma_E \rangle\\
        {}&= \lambda^2 W_E\vert_{\underline{v}} \langle 1 \rangle =0.
    \end{split}
\end{equation*}
\end{proof}

\subsection{Trajectories, Prolongations and the Horizontal Condition on \texorpdfstring{$U$}{TEXT}}\label{Subsec ICs}
The notion of integral curves for Vlasov fields on kinematic domains can be extended to integral curves on $U$.
The Vlasov field $W$ on $U$ generates a set of integral curves on $U$, denoted $\eta$.
The projections of these curves onto $M$ are exactly the trajectories of the particles and are denoted $C$.
Since we restrict our attention here to integral curves of horizontal vector fields, the terms prolongation and integral curve can be used interchangeably.

Trajectories can be considered in terms of maps from intervals of the real line $\mathcal{I}\subset \mathbb{R}$ into the spacetime manifold $M$:
\begin{equation}
    C\colon  \mathcal{I}\hookrightarrow M.
\end{equation}
These trajectories can be parameterised by choosing a parameter, $t\in\Gamma\Lambda^0\mathcal{I}$ with $dt\neq0$.
We then write the parameterised trajectories as $C(t)$.

\begin{definition}[Prolongations]\label{Def Parametrised Integral Curves}
The \textit{prolongation} of a curve $C\colon\mathcal{I}\hookrightarrow M$ is given by
\begin{equation}
\eta:\mathcal{I}\hookrightarrow U,\; \eta(t_0)=C_\ast(\partial_t\vert_{t_0}),\; \forall t_0\in \mathcal{I}.
\end{equation}
\end{definition}

\begin{lemma}
$\eta:\mathcal{I}\hookrightarrow U$ is the prolongation of some trajectory $C$ if and only if
\begin{equation}
(\pi\circ\eta)_\ast(\partial_t)=\eta.
\label{Eqn Prolongation Property}
\end{equation}
\end{lemma}
\begin{proof}

First suppose $\eta$ is the prolongation of $C$. Then $C_\ast(\partial_t|_{t_0})\in T_{C(t_0)}M$. Hence $\pi\big(\eta(t_0)\big) = \pi\big(C_\ast(\partial_t|_{t_0})\big)=C(s)$, i.e. $\pi\circ\eta=C$ and hence $(\pi\circ\eta)_\ast(\partial_t)=C_\ast(\partial_t)=\eta$.

Conversely assuming \cref{Eqn Prolongation Property}, then let $C=\pi\circ\eta$. Thus
$C_\ast(\partial_t) = (\pi\circ\eta)_\ast(\partial_t) = \eta$.
\end{proof}

Let $\hat{\mathcal{I}}\subset\mathbb{R}$ be coordinated by $\hat{t}$ and let $\hat{C}:\hat{\mathcal{I}}\hookrightarrow M$ be an alternative parameterisation of $C$. I.e. there exists a diffeomorphism $\hat{t}=\hat{t}(t)$ such that
\begin{equation}
    \hat{C}\big(\hat{t}(t)\big)=C(t).
\end{equation}
Note that although the two parameterisations define the same curve, there prologations $\eta$ and $\hat\eta$ do not coincide.
Consequently, the tangent vectors along the prolongations (i.e. the acceleration) belong to different spaces:
$\eta_\ast(\partial_t\vert_{t_0}) \in T_{\eta(t_0)}U$ and  $\hat{\eta}_\ast(\partial_{\hat{t}}\vert_{t_1}) \in T_{\hat{\eta}(t_1)}U$
even when $C(t_0)=\hat{C}(t_1)$.

\begin{definition}
An integral curve $\eta$ of $W$ is said to lie along $E$ if
\begin{equation}
\eta(t) \in E,\; \forall t\in \mathcal{I}.
\end{equation}
This is a necessary condition $\eta$ to be an integral curve of $W_E$.
\end{definition}
If $W_E$ is induced by $W$ and $\eta$ is an integral curve of $W$ and $\eta$ lies along $E$ then $\eta$ is also an integral curve of $W_E$.

\begin{lemma}
All integral curves $\eta$ of a vector field $X\in\Gamma TU$ are prolongations if and only if $X$ is horizontal. 
\end{lemma}

\begin{proof}
Let $\underline{u}\in U$ and let $\eta$ be an integral curve of $X$ i.e. $X\vert_{\eta(t_1)}= \eta_\ast(\partial_t\vert_{t_1})$ for all $t_1\in\mathcal{I}$, such that $\eta(t_0)=\underline{u}$ for some $t_0\in\mathcal{I}$.

Suppose first that $\eta$ is a prolongation.
Then
\begin{equation}
\begin{split}
    \pi_\ast X\vert_{\underline{u}} &= \pi_\ast X\vert_{\eta(t_0)} = \pi_\ast\eta_\ast (\partial_t\vert_{t_0})\\
    {}& = (\pi\circ\eta)_\ast (\partial_t\vert_{t_0})= \eta(t_0)= \underline{u}.
\end{split}
 \end{equation}
Hence $X$ is horizontal.

Suppose now that $X$ is horizontal.
We have that $\pi_\ast X\vert_{\underline{u}}= \underline{u}= \eta(t_0)$.
Then,
\begin{equation}
\eta(t_0)= \pi_\ast X\vert_{\eta(t_0)} = \pi_\ast\eta_\ast (\partial_t\vert_{t_0}) =(\pi\circ\eta)_\ast (\partial_t\vert_{t_0}).
\end{equation}
Hence $\eta$ is a prolongation.
\end{proof}

The integral curves $\eta$ of Vlasov fields $W\in \Gamma TU$ correspond to trajectories $C$ in the base space $M$ through the following relation:
\begin{equation}\label{Eqn Trajectories}
C=\pi\circ\eta.
\end{equation}
In local coordinates $(x^\mu,\dot{x}^\mu)$ 
let $C^\mu(t_0)=x^\mu\vert_{C(t_0)}$ then
\begin{align}
\dot{C}^\mu(t_0)= \frac{dC^\mu}{dt} \bigg{\vert}_{t_0} \; \mbox{ and } \;
\ddot{C}^\mu(t_0)= \frac{d^2C^\mu}{dt^2}\bigg{\vert}_{t_0}.
\end{align}

These trajectories can be expressed in terms of a parametrised system of second order differential equations in terms of the coefficients of $W$:
\begin{equation}\label{Eqn Trajectory ODE}
\ddot{C}^\mu\vert_{t_0}=\varphi^\mu\vert_{\dot{C}(t_0)},
\end{equation}
where $t$ is a parameter which corresponds to $W$.
Due to the projection in \cref{Eqn Trajectories}, there is a class of integral curves that produce the same trajectories.
The vector fields which produce these curves differ only by a term $k\mathcal{R}$ where $k\in\Gamma\Lambda^0U$ is a 1--homogeneous function.
The converse is also true and a full proof of this statement can be found in \cite{shen_differential_2001}.
These vector fields are said to be projectively related in the literature (see \cref{Def Projectively Related} for a definition).
We present here a proof of only the former statement to illustrate this result in our mathematical notation.

\begin{lemma}\label{Lem Trajectories}
Let $W,\hat{W}\in\Gamma TU$ be Vlasov fields and let $k\in\Gamma\Lambda^0U$ be a 1--homogeneous scalar field. 
If $W$ and $\hat{W}$ are related by 
\begin{equation}\label{Eqn W and W hat}
\hat{W}=W+k\mathcal{R}
\end{equation}
then they have the same trajectories up to a parameterisation.
That is, if the trajectories $C$ of $W$ are parametrised by $t$ then $\hat{W}$ has the same trajectories $C$ only parametrised by $s$ where
\begin{equation}\label{Eqn Reparameterisation}
\frac{d^2s}{dt^2}\bigg{\vert}_{t_0}+ k\vert_{\underline{u}} \frac{ds}{dt}\bigg{\vert}_{t_0}=0,\; \frac{ds}{dt}>0,
\end{equation}
and $\underline{u}=C_\ast(\partial_t\vert_{t_0})\in U$.
\end{lemma}

\begin{proof}
Notice that if $W$ and $\hat{W}$ are related by \cref{Eqn W and W hat} then in local coordinates we have
\begin{equation*}
\hat{\varphi}^\mu=\varphi^\mu+ k\dot{x}^\mu,
\end{equation*}
where $\varphi^\mu$ (resp. $\hat{\varphi}^\mu$) are the $\partial_\mu^{(\dot{x})}$ coefficients of $W$ ($\hat{W}$), see \cref{Vlasov Coords}.
Let the trajectories of $W$ be denoted by $C(t)$.
These trajectories satisfy \cref{Eqn Trajectory ODE}.
Define a new parameter $s=s(t)$ by \cref{Eqn Reparameterisation} and set $s_0=s(t_0)$.

The trajectory $C(s(t))$ can be shown to satisfy
\begin{equation*}
\frac{d^2C^\mu}{dt^2}\bigg{\vert}_{t_0}= \frac{d^2s}{dt^2}\bigg{\vert}_{t_0}\frac{dC^\mu}{ds}\bigg{\vert}_{s_0} +\left( \frac{ds}{dt}\bigg{\vert}_{t_0} \right)^2 \frac{d^2C^\mu}{ds^2}\bigg{\vert}_{s_0}.
\end{equation*}
We therefore have
\begin{equation*}
\begin{split}
\frac{d^2C^\mu}{ds^2}&\bigg{\vert}_{s_0}= \left( \frac{ds}{dt}\bigg{\vert}_{t_0} \right)^{-2} \! \left( \frac{d^2C^\mu}{dt^2}\bigg{\vert}_{t_0}- \frac{d^2s}{dt^2}\bigg{\vert}_{t_0} \frac{dC^\mu}{ds}\bigg{\vert}_{s_0} \right)\\
{}&= \left( \frac{ds}{dt}\bigg{\vert}_{t_0} \right)^{-2}  \bigg{(} \varphi^\mu\vert_{C_\ast(\partial_t\vert_{t_0})}\\ 
{}& \hspace{8em} +k\vert_{C_\ast(\partial_t\vert_{t_0})} \frac{ds}{dt}\bigg{\vert}_{t_0} \frac{dC^\mu}{ds}\bigg{\vert}_{s_0} \bigg{)}\\
{}&= \varphi^\mu\vert_{C_\ast(\partial_s\vert_{s_0})}+ k\vert_{C_\ast(\partial_s\vert_{s_0})} \frac{dC^\mu}{ds}\bigg{\vert}_{s_0}\\
{}&= \left( \varphi^\mu+k\dot{x}^\mu \right)\vert_{C_\ast(\partial_s\vert_{s_0})}\\
{}&= \hat{\varphi}^\mu\vert_{C_\ast(\partial_s\vert_{s_0})}.
\end{split}
\end{equation*}
The third line is due to the 2-homogeneity of $\varphi^\mu$ and 1-homogeneity of $k$.
Hence $C(s)$ are the trajectories associated with $\hat{W}$.
Hence, both $W$ and $\hat{W}$ have the same trajectories up to a reparameterisation.
\end{proof}

\subsection{Transforming between Kinematic Domains}\label{Sec Transforming Between Domains}
Suppose we are given a Vlasov field $W_E$ on a kinematic domain $E$ with \Fkin $F$.
Using this data we can construct a new Vlasov field $\hat{W}$ which is compatible with a new kinematic domain $\hat{E}$ with \Fkin $\hat{F}$.
This defines a new Vlasov field $\hat{W}_{\hat{E}}$ on $\hat{E}$ which produces the same trajectories in the base space $M$ as the initial Vlasov field $W_E$.

Once we have promoted $W_E$ into a Vlasov field $W$ on $U$ using \cref{Lem WE to W}, we can use it to construct a new Vlasov field $\hat{W}$ which is tangent to another kinematic domain $\hat{E}$.
An illustration of the proceeding lemma can be found in \cref{Magic Formula Diagram}.

\begin{theorem}\label{Thm Magic Formula}
Let $E$ and $\hat{E}$ be kinematic domains with \Fkin $F$ and $\hat{F}$ respectively.
Given a Vlasov field $W\in\Gamma TU$ which is compatible with $E$ we may construct a new Vlasov field $\hat{W}\in\Gamma TU$ given by
\begin{equation}\label{Magic Formula}
\hat{W}=W-\frac{W\langle \hat{F}\rangle}{\mathcal{R}\langle \hat{F}\rangle} \mathcal{R}
\end{equation}
which is compatible with $\hat{E}$ and corresponds to the same trajectories as $W$ (up to parameterisation).
\end{theorem}

\begin{proof}
It is clear that $\hat{W}\langle \hat{F}\rangle=0$.
To see that the resultant vector field is horizontal observe that
\begin{equation*}
\begin{split}
\pi_\ast(\hat{W}\vert_{\underline{u}})& = \pi_\ast(W\vert_{\underline{u}}) -\frac{W\langle \hat{F}\rangle}{\mathcal{R}\langle \hat{F}\rangle} \bigg{\vert}_{\underline{u}} \pi_\ast(\mathcal{R}\vert_{\underline{u}})\\
{}&=  \pi_\ast(W\vert_{\underline{u}})= \underline{u}
\end{split}
\end{equation*}

To see that $\hat{W}$ is radially quadratic observe that
\begin{equation*}
\begin{split}
\hat{W}\vert_{\lambda \underline{v}} \langle \dot{h}\rangle=& W\vert_{\lambda \underline{v}} \langle \dot{h}\rangle- \frac{W\langle \hat{F}\rangle}{\mathcal{R}\langle \hat{F}\rangle} \bigg{\vert}_{\lambda\underline{v}}  \mathcal{R}\vert_{\lambda \underline{v}} \langle \dot{h}\rangle\\
=& \lambda^2 W\vert_{\underline{v}} \langle \dot{h}\rangle- \frac{\lambda^{k+1}W\vert_{\underline{v}}\langle \hat{F}\rangle }{\lambda^k \mathcal{R}\vert_{\underline{v}} \langle \hat{F}\rangle}   (\lambda \underline{v} \langle \dot{h}\rangle)\\
=& \lambda^2 \bigg{(} W\vert_{\underline{v}} \langle \dot{h}\rangle- \frac{W \langle \hat{F}\rangle }{ \mathcal{R} \langle \hat{F}\rangle}\bigg{\vert}_{\underline{v}} \mathcal{R}\vert_{\underline{v}} \langle \dot{h}\rangle  \bigg{)}\\
=& \lambda^2 \hat{W}\vert_{\underline{v}} \langle \dot{h} \rangle,
\end{split}
\end{equation*}
for any $h\in\Gamma\Lambda^0 M$, $\underline{v}\in U$, $\lambda\neq0$.
Hence $\hat{W}$ is a valid Vlasov field.
To see that $\hat{W}$ corresponds to the same set of trajectories as $W$, notice that $W\langle \hat{F}\rangle / \mathcal{R}\langle \hat{F}\rangle $ in \cref{Magic Formula} is a 1--homogeneous function.
Hence by \cref{Lem Trajectories}, $W$ and $\hat{W}$ have the same trajectories up to parameterisation.
\end{proof}

Recall that \cref{Eqn LT Vlasov} consisted of three terms: the first two can be identified with the terms in \cref{Eqn LF Vlasov}, while the third can be show to be
\begin{equation}
    \begin{split}
        \frac{W_{\EH}\left\langle \dot{t} \right\rangle }{\mathcal{R}\left\langle \dot{t} \right\rangle }\mathcal{R} 
    = & \frac{q}{m} \sqrt{\FkinUH} g^{\lambda\nu}\mathcal{F}_{\nu\rho}\dot{x}^\rho \frac{\partial t}{\partial x^\lambda} \frac{\dot{x}^\mu}{\dot{t}}\\
    {}&- \Gamma^\lambda_{\nu\rho}\dot{x}^\nu\dot{x}^\rho \frac{\partial t}{\partial x^\lambda} \frac{\dot{x}^\mu}{\dot{t}}\\
    {}& +\dot{x}^\nu\dot{x}^\rho \frac{\partial t}{\partial x^\nu\partial x^\rho} \frac{\dot{x}^\mu}{\dot{t}}.
    \end{split}
\end{equation}
This is exactly the additional term from \cref{Magic Formula}.

\subsection{The Transport Equations on Kinematic Domains}\label{Subsec Transport Equations}
A method of interpreting the Vlasov equation in a kinematic domain $E$ can be given in terms of the transport equations \cite{gratus_distributional_2010}.
The transport equations are written in terms of a particle density form which codifies the phase space trajectories corresponding to a given Vlasov field.
With some additional structure (a choice of non-vanishing top form, otherwise known as a measure) we can recover the usual Vlasov equation (i.e. a particle density function $f_E\in\Gamma\Lambda^0 E$ such that $W_E\langle f_E \rangle =0$) from the transport equations on our preferred choice of kinetic domain $E$. 

In section \cref{Subsec Current Density}, we show how to calculate the current on $M$, and show it is independent of the choice of $E$. While in section \cref{subsect Stress energy}, we give the formula for the stress-energy tensor, which does depend on $E$. 

\begin{definition}[Particle Density Form]
    A \textit{particle density form} on $E$, $\theta_E\in\Gamma\Lambda^{2n-2}E$ is a $2n-2$--form that satisfies the transport equations. 
\end{definition}

\begin{definition}[Transport Equations on $E$]\label{Def Transport Equations E}
Consider a kinetic domain $E\subset U$, a particle density 6--form $\theta_E\in\Gamma\Lambda^{2n-2}E$, and a Vlasov field $W_E\in\Gamma TE$.
The \textit{transport equations on} $E$ are written
\begin{equation}
d\theta_E=0,\; i_{W_E}\theta_E=0.
\label{Eqn Transport Equations E}
\end{equation}
A visualisation of the transport equations is given in \cref{Fig Transport Equations}.
\end{definition}

Let $E$ be a kinematic domain, $W_E\in\Gamma TE$ a Vlasov field on $E$, and let $\Omega_E\in\Gamma\Lambda^{2n-1}E$ be a measure on $E$ such that
\begin{equation}\label{Eqn LWE Omega =0}
L_{W_E}\Omega_E=0.
\end{equation}
We can relate a particle density $(2n{-}2)$--form $\theta_E\in\Gamma\Lambda^{2n-2}E$ and a particle density function $f_E\in\Gamma\Lambda^0E$ via
\begin{equation}\label{Eqn Theta E I}
\theta_E= f_Ei_{W_E} \Omega_E.
\end{equation}

\begin{lemma}
Given $\theta_E\in\Gamma\Lambda^{2n-2}E$ such that $i_{W_E}\theta_E=0$ then $f_E\in\Gamma\Lambda^0E$, given by \cref{Eqn Theta E I} exists and is unique. 
Furthermore, $\theta_E$ satisfies the transport equations \cref{Eqn Transport Equations E} if and only if $f_E$ satisfies \cref{Intr Vlasov Eqn on E}, that is $W_E\langle f_E\rangle =0$.
\end{lemma}

\begin{proof}
Let $(x^0,...,x^{2n-2})$ define a coordinate system on $E$ such that $\partial_0^{(x)}=W_E$.
Then,
\begin{equation*}
    \Omega_E=\Omega_0dx^0\wedge \cdots \wedge dx^{2n-2}, \mbox{ and } 
    \theta_E= \theta_E^\mu i_\mu^{(x)}\Omega_E.
\end{equation*}
We then have
\begin{equation*}
    \begin{split}
        i_{W_E}\theta_E &= i_0^{(x)}\theta_E\\
        {}& = \theta_E^0 i_0^{(x)}i_0^{(x)} \Omega + \theta_E^a i_0^{(x)} i_a^{(x)} \Omega\\
        {}&=0.
    \end{split}
\end{equation*}
Hence $\theta_E^a=0$. 
It follows that $\theta_E= \theta_E^0i_0^{(x)}\Omega$ i.e. $\theta_E=f_E i_{W_E}\theta_E$.

From \cref{Eqn Theta E I}, we have
\begin{equation*}
\begin{split}
d\theta_E=& df_E\wedge i_{W_E}\Omega_E= i_{W_E}df \wedge \Omega_E\\
=& W_E\langle f_E\rangle \Omega_E, 
\end{split}
\end{equation*}
hence the transport equations for $\theta_E$ are equivalent to \cref{Intr Vlasov Eqn on E}.
\end{proof}

An example of the transport equations can be seen in \cite{hohmann_non-metric_2016}.
These equations however are subtly different in that a form $\omega_E=i_{W_E}\Omega_E$ with $L_{W_E}\Omega_E =0$ is defined and the transport equations are cast as $d\omega_E=0$ and $i_{W_E}\omega_E=0$.
Note that the particle density function $f_E$ is absent from this definition.

Given a volume form $\Omega\in\Gamma\Lambda^{2n} U$ and kinematic domain $E\subset U$ with the map $\Sigma_E: E\hookrightarrow U$, the typical choice of top form on $E$ is given by 
\begin{equation}\label{Eqn Omega E from Omega}
\Omega_E= \Sigma_E^\ast (i_{\mathcal{R}}\Omega).
\end{equation}

\begin{lemma}
    Let $\Omega\in\Gamma\Lambda^nU$ such that
    $L_W\Omega=0$. Let $\Omega_E\in\Gamma\Lambda^{n-1}E$ be defined by \cref{Eqn Omega E from Omega} then
    \cref{Eqn LWE Omega =0} holds.
\end{lemma}

\begin{proof}
    Observe that
    \begin{equation*}
        \begin{split}
        L_{W_E}\Omega_E&= L_{W_E} \left( \Sigma_E^\ast i_{\mathcal{R}}\Omega\right)
        = \Sigma_E^\ast\left( L_Wi_{\mathcal{R}}\Omega\right)\\
        {}&= \Sigma_E^\ast\left(i_{[W,\mathcal{R}]}\Omega + i_{\mathcal{R}} L_W\Omega \right)
        = \Sigma_E^\ast i_W\Omega\\
        {}&= i_{W_E} \Sigma_E^\ast \Omega = 0,
        \end{split}
    \end{equation*}
    since the degree of $\Omega$ is greater than the dimension of $E$.
\end{proof}

In the instance that the chosen kinematic domain is the upper unit hyperboloid $\EH$ as given in \cref{Unit Hyperboloid Bundle}, then the volume form typically chosen is
\begin{equation}\label{Eqn Vol EH}
\Omega_{\EH}= -\frac{\det (g)}{\dot{x}_0} dx^0\wedge\cdots\wedge dx^{n-1}\wedge d\dot{x}^1\wedge\cdots\wedge d\dot{x}^{n-1}.
\end{equation}
Note that a volume form on $E\subset U$ can be constructed by choosing a volume form on $\Omega\in\Gamma\Lambda^{2n}U$ and then pulled back onto $E$.
Furthermore, if $W_{\EH}$ is the Vlasov field corresponding to the Lorentz force it can be shown that $L_{W_{\EH}}\Omega_{\EH}=0$ and hence the transport equations equipped with $\Omega_E$ are equivalent to the standard Vlasov equation.

There are instances however when it is advantageous to cast the Vlasov equation in terms of the transport equations even when the flow of the measure is not preserved by the Vlasov field (i.e. $L_{W_E}\Omega_E\neq 0$).
In \cite{noble_kinetic_2013} an attempt is made to account for the radiation reaction within the Vlasov field over a 4-dimensional spacetime manifold $M$.
To accomplish this, the Vlasov equation is cast as 
\begin{equation}
L_{W_Q}(f_Q\omega_Q)=0.
\end{equation}
Here the 10-dimensional manifold $Q$ is a subset of the double copy of the tangent bundle $\subset TM \oplus TM$ where the first copy of the tangent bundle contains normalised timelike velocity vectors (as on the upper unit hyperboloid), and the second copy contains acceleration vectors which are orthogonal to these velocities.
The measure $\omega_Q\in\Gamma\Lambda^{10}Q$ is formed on $Q$ by pulling back a specific 10 form related to the measure on $TM \oplus TM$.
In this case, the flow of $\omega_Q$ along $W_Q$ is not conserved:
\begin{equation}
L_{W_Q}\omega_Q= \frac{3}{\tau} \omega_Q, \mbox{ or equivalently } W_Q \langle f_Q\rangle + \frac{3}{\tau}f_Q=0,
\end{equation}
where $\tau=q^2/6\pi m$.
The latter equation is similar to \cref{Intr Vlasov Eqn on E}, but the additional term on the LHS is to account for losses due to radiation.
Note that, although the Vlasov equation is altered, the transport equations remain the same:
\begin{equation}
    d \theta_Q=0, \quad i_{W_Q}\theta_Q=0,
\end{equation}
where
\begin{equation}
    \theta_Q= i_{W_Q}(f_Q\omega_Q).
\end{equation}


\subsection{Example: The unit hyperboloid and lab time Vlasov vector fields for the Lorentz force equation.}
\label{subsec Vlasov LF}

In \cref{Eqn LF Vlasov}, \cref{Eqn LT Vlasov} and \cref{Eqn LT Vlasov Coord} we gave the Vlasov vector field adapted to the upper unit hyperboloid, an arbitrary time slicing and a time slicing with respect to coordinate time. We stated that one method to find the transformation is to calculate 2nd order ODES, perform the reparameterisation, and then recalculate the corresponding Vlasov vector field, which we do here. 

\begin{lemma}
    Equation \cref{Eqn LF Vlasov} can be transformed into equation \cref{Eqn LT Vlasov}, by reparameterising the Lorentz force ODE. 
\end{lemma}

\begin{proof}
Let $t$ be the parameterisation of the trajectories $C$ of the Vlasov field $W_{\EH}$ (i.e. \cref{Eqn LF Vlasov}.
The Lorentz force ODE is given in coordinates by
\begin{equation*}
    \frac{d^2 C^\mu}{d\tau^2}
    = \frac{q}{m}\sigma\sqrt{\FkinUH} g^{\mu\nu}\mathcal{F}_{\nu\rho} \frac{dC^\rho}{d\tau}
    - \Gamma^\mu_{\nu\rho} \frac{dC^\nu}{d\tau} \frac{dC^\rho}{d\tau}.
\end{equation*}
Let $W_{\ELab{t}}$ (i.e. \cref{Eqn LT Vlasov}) be parametrised by $s$.
Assuming the existence of a relationship $t=t(\tau)$ the above ODE can be reparametrised as
\begin{equation}\label{Eqn Reparametrised LF ODE}
     \begin{split}
         \frac{d^2 C^\mu}{d   t^2}
     =& \frac{q}{m}\sigma\sqrt{\FkinUH} g^{\mu\nu}\mathcal{F}_{\nu\rho} \frac{dC^\rho}{d t} 
     - \Gamma^\mu_{\nu\rho} \frac{dC^\nu}{d t} \frac{dC^\rho}{d t}\\
     {}&- \left( \frac{dt}{d\tau} \right)^{-3} \frac{d^2t}{d\tau^2} \frac{dC^\mu}{dt}.
     \end{split}
\end{equation}
Let the integral curves of $W_{\EH}$ and $W_{\ELab{t}}$ be given by $\eta_{\EH}$ and $\eta_{\ELab{t}}$ respectively.
Let $t_0=t(\tau_0)$ then we have the following relations:
\begin{align*}
    \frac{dt}{d\tau} \bigg{\vert}_{\tau_0}=& \dot{t}\vert_{\eta_{\EH}(\tau_0)},
\end{align*}
\begin{equation*}
    \begin{split}
        \frac{d^2t}{d\tau^2}& \bigg{\vert}_{\tau_0}
        = \eta_{\EH}{}_\ast(\partial_\tau\vert_{\tau_0}) \langle \dot{t}\rangle
        = W_{\EH}\vert_{\eta_{\EH}(\tau_0)} \langle \dot{t}\rangle\\
        =&  W_{\EH}\vert_{\eta_{\EH}(\tau_0)}\left\langle \dot{x}^\nu \frac{\partial t}{\partial x^\nu} \right\rangle \\
        =& \bigg{(}  \frac{q}{m} \sigma \sqrt{\FkinUH} g^{\lambda\nu}\mathcal{F}_{\nu\rho}\dot{x}^\rho \frac{\partial t}{\partial x^\lambda}\\ 
        {}&- \Gamma^\lambda_{\nu\rho}\dot{x}^\nu\dot{x}^\rho \frac{\partial t}{\partial x^\lambda} 
        +\dot{x}^\nu\dot{x}^\rho \frac{\partial t}{\partial x^\nu\partial x^\rho}   \bigg{)} \bigg{\vert}_{\eta_{\EH}(\tau_0)}\\
        =& \frac{q}{m} \left( \sigma \sqrt{\FkinUH} g^{\lambda\nu}\mathcal{F}_{\nu\rho} \right)\big{\vert}_{\eta_{\EH}(\tau_0)} \frac{d C^\rho}{d\tau}\bigg{\vert}_{\tau_0} \frac{\partial t}{\partial x^\lambda}\bigg{\vert}_{\tau_0}\\ 
        {}&- \Gamma^\lambda_{\nu\rho} \big{\vert}_{\eta_{\EH}(\tau_0)} \frac{d C^\nu}{d\tau}\bigg{\vert}_{\tau_0} \frac{d C^\rho}{d\tau}\bigg{\vert}_{\tau_0} \frac{\partial t}{\partial x^\lambda} \bigg{\vert}_{\tau_0}\\
         {}& +\frac{d C^\nu}{d\tau}\bigg{\vert}_{\tau_0} \frac{d C^\rho}{d\tau}\bigg{\vert}_{\tau_0} \frac{\partial t}{\partial x^\nu\partial x^\rho}\bigg{\vert}_{\tau_0}\\
        =& \left( \frac{dt}{d\tau}\bigg{\vert}_{\tau_0} \right)^2 \bigg{(}
        \frac{q}{m} \left( \sigma \sqrt{\FkinUH} g^{\lambda\nu}\mathcal{F}_{\nu\rho} \right)\big{\vert}_{\eta_{\ELab{t}}(t_0)} \frac{d C^\rho}{dt}\bigg{\vert}_{t_0} \frac{\partial t}{\partial x^\lambda}\bigg{\vert}_{t_0} \\
        {}& - \Gamma^\lambda_{\nu\rho} \big{\vert}_{\eta_{\ELab{t}}(t_0)} \frac{d C^\nu}{dt}\bigg{\vert}_{t_0} \frac{d C^\rho}{dt}\bigg{\vert}_{t_0} \frac{\partial t}{\partial x^\lambda} \bigg{\vert}_{t_0}\\
        {}&\quad +\frac{d C^\nu}{dt}\bigg{\vert}_{t_0} \frac{d C^\rho}{dt}\bigg{\vert}_{t_0} \frac{\partial t}{\partial x^\nu\partial x^\rho}\bigg{\vert}_{t_0}  \bigg{)}
    \end{split}
\end{equation*}
Inserting the above relationships into \cref{Eqn Reparametrised LF ODE} we get
\begin{equation}\label{Eqn Reparametrised LF ODE II}
    \begin{split}
        \frac{d^2 C^\mu}{d   t^2}
     =& \frac{q}{m}\sigma\sqrt{\FkinUH} g^{\mu\nu}\mathcal{F}_{\nu\rho} \frac{dC^\rho}{d t} 
     - \Gamma^\mu_{\nu\rho} \frac{dC^\nu}{d t} \frac{dC^\rho}{d t}\\
     {}& -\left( \dot{t}\right)^{-1}  \bigg{(}
        \frac{q}{m}  \sigma \sqrt{\FkinUH} g^{\lambda\nu}\mathcal{F}_{\nu\rho}  \frac{d C^\rho}{dt} \frac{\partial t}{\partial x^\lambda}\\ 
    {}& \qquad - \Gamma^\lambda_{\nu\rho} \frac{d C^\nu}{dt} \frac{d C^\rho}{dt} \frac{\partial t}{\partial x^\lambda}\\ 
    {}& \qquad +\frac{d C^\nu}{dt} \frac{d C^\rho}{dt} \frac{\partial t}{\partial x^\nu\partial x^\rho}  \bigg{)} \frac{dC^\mu}{dt}.
    \end{split}
\end{equation}

We may identify this with the $\partial^{(\dot{x})}_\mu$ term in \cref{Eqn LT Vlasov} using \cref{Eqn Trajectory ODE}.
Hence $W_{\ELab{t}}$ corresponds to the reparametrised Lorentz force equation above.
If we were to choose an adapted coordinate system $(x^0=t,x^a,\dot{x}^\mu)$ then we would have
\begin{equation}
    \frac{\partial t}{\partial x^\mu}=\delta^0_\mu, \quad \frac{\partial t}{\partial x^\mu \partial^\nu}=0.
\end{equation}
Plugging these into \cref{Eqn Reparametrised LF ODE II} gives the $\partial^{(\dot{x})}_\mu$ coefficient of \cref{Eqn LT Vlasov Coord}.
\end{proof}

Alternatively we can use the transformation formula \cref{Magic Formula}. 
\begin{lemma}
    Equation \cref{Eqn LF Vlasov} can be transformed into equations \cref{Eqn LT Vlasov} and hence \cref{Eqn LT Vlasov Coord}, using the transformation formula \cref{Magic Formula}. 
\end{lemma}
\begin{proof}
    This is just a trivial calculation.
\end{proof}

Thus we see that the use of the transformation formula greatly simplifies the task of re-adapting the Vlasov vector field from one kinematic domain to another. Furthermore we will see below, in  that equation \cref{Magic Formula} can be trivially derived from the Vlasov bivector. This is shown in \cref{Lem MAgic Fromula From Psi} below.


\section{The Vlasov Bivectors}\label{Sec Vlasov Bivectors}
Here we present the Vlasov Bivector, the key goal of this article. The advantages of this approach have been listed in section \cref{Sec: Intro Advantages}. 
The Vlasov bivector $\Psi$ can always be expressed as 
\begin{equation}\label{Vlasov Bivector}
\Psi= \mathcal{R} \wedge W,
\end{equation}
for any appropriate $W\in\Gamma TU$.
It is trivial to see that any Vlasov fields related by \cref{Eqn W and W hat} and \cref{Magic Formula} produce the same Vlasov bivector, and so represents an entire class of equivalent Vlasov fields.
Throughout this section, $\Psi$ will be reserved for Vlasov bivectors while arbitrary bivectors will be denoted $\Phi$.

In \cref{Subsect Bivectors} we present the necessary results for arbitrary bivectors, not necessarily the Vlasov Bivector. This is followed in \cref{Subsect Horizontal Bivectors} with more conditions on a general bivector, until in \cref{Subsect Vlasov Bivectors} we impose all the conditions required on bivector to be Vlasov.

\subsection{Bivectors}\label{Subsect Bivectors}

\begin{definition}[Bivector]
A \textit{bivector} over a manifold $N$ is an exterior product of vector fields $X_\mu,Y_\mu\in\Gamma\ TN$ over $N$, denoted by $\sum_\mu X_\mu\wedge Y_\mu$. Here $\mu$ is taken to be an arbitrary summation index. These follow the standard rules for exterior products, namely `f'--linearity and antisymmetry.
\end{definition}

\begin{definition}[Simple Bivectors]
A bivector over $N$, $\Phi$ is said to be \textit{simple} if there exists $X,Y\in\Gamma TN$ such that $\Phi=X\wedge Y$.
The space of all simple bivectors (fields) over $N$ is denoted by $\Gamma\mathcal{B}^2(N)$.
\end{definition}

We require that a Vlasov bivector be simple and take the form \cref{Vlasov Bivector}.
In order to formally write the definition of $\Psi$ we first observe the following property of simple bivectors.

\begin{definition}[Bivector Pairing]
A bivector $\sum_\mu X_\mu \wedge Y_\mu$ acts on a pair of scalar fields $F,G\in\Gamma\Lambda^0U$ according to
\begin{equation}
\sum_\mu\big(X_\mu\wedge Y_\mu)\langle F,G\rangle= \sum_\mu( X_\mu\langle F\rangle Y_\mu\langle G\rangle- X_\mu\langle G\rangle Y_\mu\langle F\rangle \big).
\end{equation}
Similarly we may define a vector field using
\begin{equation}
\sum_\mu(X_\mu\wedge Y_\mu)\langle F,\bullet\rangle= \sum_\mu\big( X_\mu\langle F\rangle Y_\mu- \ Y_\mu\langle F\rangle X_\mu\big).
\end{equation}
In the case that we have a simple bivector acting on a scalar field we write
\begin{equation}\label{Eqn Biv on Scalar}
(X\wedge Y)\langle F,\bullet \rangle = X\langle F\rangle Y- Y\langle F\rangle X.
\end{equation}
\end{definition}

\begin{lemma}\label{Lem X wedge Phi}
Given a bivector $\Phi$ and a non-zero vector field $X\in\Gamma TN$, $X\wedge\Phi=0$ if and only if there exists some $Y\in\Gamma\ TN$ such that $\Phi=X\wedge Y$.
\end{lemma}

\begin{proof}
Suppose that $X\wedge \Phi=0$.
Define a local coordinate system such that $X=\partial_0$ when we may write $\Phi=\Phi^{0a}\partial_0\wedge \partial_a+\frac{1}{2} \Phi^{bc}\partial_b \wedge \partial_c$.
Since $\partial_0\wedge \Phi= \frac{1}{2} \Phi^{bc}\partial_0\wedge\partial_b\wedge\partial_c=0$, we must have that each $\Phi^{bc}=0$.
Hence $\Phi= \partial_0\wedge (\Phi^{0a}\partial_a)=X\wedge Y$.
The converse holds by the properties of the exterior product ($X\wedge X=0$).
\end{proof}

\begin{definition}[Specially Related Pairs of Vector Fields] \label{Linearly Related}
Two pairs of vector fields $X_1,X_2\in \Gamma TN$ and $Y_1,Y_2\in \Gamma TN$ are said to be \textit{specially related} if there exists $\alpha,\beta,\gamma,\delta\in\Gamma\Lambda^0N$ satisfying
\begin{equation*}
( \alpha\delta- \beta\gamma )\vert_p=1\; \forall p\in N,
\end{equation*}
such that
\begin{equation*}
\begin{split}
Y_1=&\alpha X_1+\beta X_2,\\
Y_2=&\gamma X_1+\delta X_2.
\end{split}
\end{equation*}
\end{definition}

\begin{lemma}\label{Lem Lin Rel Lemma}
Two pairs of vector fields, $X_1,X_2\in\Gamma TN$ and $Y_1,Y_2\in\Gamma TN$, are specially related if and only if $X_1\wedge X_2= Y_1\wedge Y_2$. 
\end{lemma}

\begin{proof}
See \cref{Sect Proofs}.
\end{proof}

For a Vlasov bivector to contain all the necessary information to define a system, it must satisfy three conditions: the radial condition, the horizontal condition, and a third condition. There are three equivalent ways to express this third conditions: integrability, being radially cubic and being expressible as in \cref{Vlasov Bivector}. 

We begin by discussing integrability.
A geometric interpretation of a bivector is a network of infinitesimal rectangles whose sides are defined by a pair of vectors.
When these bivectors `knit together' to form smooth surfaces, the vector-distribution spanned by the components of the bivector is integrable. These surfaces are depicted in figure \cref{Geometric Bivector}.
More formally, we may consider integrability using the Frobenius theorem, through the language of vector distributions (a method of smoothly assigning vector subspaces of $T_xM$ to each point $x\in M$).

\begin{definition}[Tangent Bivector]\label{Def Tangent Bivector}
A bivector $\Phi$ is tangent to a surface $K\subset N$ if, given a representation for the bivector $\Phi=X\wedge Y, \; X,Y\in\Gamma TN$, both $X$ and $Y$ are tangent to $K$ in the sense described in \cref{Def Tangent}.
\end{definition}

From \cref{Lem Lin Rel Lemma} it is clear that \cref{Def Tangent Bivector} is independent of the representation of the bivector: if $\Phi=X\wedge Y$ and $\Phi=Z\wedge V$ and $X$ and $Y$ are tangent to $K$, then by \cref{Lem Lin Rel Lemma} $Z$ and $V$ are linear combinations of $X$ and $Y$ and hence tangent to $K$.

\begin{definition}[Integrability]\label{Def Integrability}
A simple bivector $\Phi\in\Gamma\mathcal{B}^2(N)$ is said to be \textit{integrable} if there exists $X,Y\in\Gamma TN$ and $\alpha,\beta\in\Gamma\Lambda^0N$ such that if $\Phi=X\wedge Y$ then $[X,Y]=\alpha X+ \beta Y$. 
By \cref{Lem Lin Rel Lemma}, if one representation of $\Phi$ satisfies the integrability condition, then so do all other representations of $\Phi$.
\end{definition}

\begin{lemma}\label{Lem Lin Rel == Integrable}
If a bivector $\Phi\in\Gamma\mathcal{B}^2(N)$ is integrable then for any representation $\Phi=X\wedge Y,\; X,Y\in\Gamma TN$, we have
\begin{equation}
[X,Y]=\alpha X+ \beta Y,
\end{equation}
for some $\alpha,\beta\in\Gamma\Lambda^0N$.
\end{lemma}

\begin{proof}
See \cref{Prf 2}.
\end{proof}

It follows from \cref{Lem Lin Rel Lemma} and \cref{Lem Lin Rel == Integrable} that integrability is a well defined property. 
That is, given a simple bivector, if one representation satisfies the integrability condition, then all representations do.\\

The Frobenius theorem states that a vector distribution over a manifold $M$ is integrable if and only if the Lie bracket of any two vectors within the distribution also lies within the distribution.
An integrable vector distribution then admits a collection of maximal connected integral manifolds which form a foliation of $M$.
Given an integrable bivector $\Phi=X\wedge Y$ we can form a vector distribution over $U$ which is spanned by $X$ and $Y$.
This generates a 2-dimensional foliation of $U$ such that $X$ and $Y$ are tangent to the leaves of our foliation at each point.
We may identify the leaves of this foliation with the 

We also introduce the null condition here.
This is a property an arbitrary bivector may have which is necessary for defining the transport equations on $U$.

\begin{definition}[Null condition] \label{Def Null Condition}
Given a $(2n{-}2)$--form $\alpha$, and a bivector $\Phi\in\Gamma\mathcal{B}^2(N)$, then \textit{the null condition} is given by
\begin{equation}\label{Eqn Null Condition}
\begin{gathered}
\mbox{Null}(\Phi,\alpha) \text{ is true}\iff  i_X\alpha=0\;\text{and}\; i_Y\alpha=0,
\\
\mbox{ for any } X,Y\in\Gamma TN \mbox{ such that } \Phi= X \wedge Y.
\end{gathered}
\end{equation}
\end{definition}

To see that the null condition is well defined notice that if $X\wedge Y= Z\wedge V$ then $\mbox{Null}(X\wedge Y,\alpha)$ holds if and only if $\mbox{Null }(Z\wedge V,\alpha)$ holds.
This is due to \cref{Lem Lin Rel Lemma} and the linearity of the contraction mapping.

\subsection{Horizontal Bivectors}\label{Subsect Horizontal Bivectors}

In order for a bivector to properly define a system we require it to satisfy certain properties.
The first of which are the radial condition and the horizontal condition.
We deal with both of these conditions in tandem.
We now work exclusively work on the submanifold $U\subset \TMSlit$.

\begin{definition}[Radial Bivectors]
A simple bivector $\Phi\in\Gamma\mathcal{B}^2(U)$ is called \textit{radial} if
\begin{equation}
\mathcal{R}\wedge\Phi=0,
\end{equation}
where $\mathcal{R}\in\Gamma TU$ is the radial vector field.
\end{definition}

\begin{definition}[Horizontal Bivectors]
A radial simple bivector $\Phi\in\Gamma\mathcal{B}^2(U)$ is called \textit{horizontal} if for any $f,h\in\Gamma\Lambda^0U$ it satisfies
\begin{equation}
\Phi\langle \pi^\ast f, \dot{h}\rangle = -\dot{f}\dot{h}.
\end{equation}
The space of horizontal bivectors is denoted $\Gamma\mathcal{B}^2_H(U)$.
\end{definition}

\begin{lemma}\label{Psi Structure Thm}
A bivector $\Phi=\mathcal{R}\wedge X$ is horizontal if and only if $X\in\Gamma TU$ is horizontal. 
\end{lemma}

\begin{proof}
Suppose first that $\Phi= \mathcal{R}\wedge X$ for some horizontal $X\in\Gamma TU$. We then have $\mathcal{R}\wedge \Psi=0$ automatically, hence $\Phi$ is radial.
To see that $\Phi$ is horizontal observe that
\begin{equation*}
\begin{split}
\Phi\langle \pi^\ast f,\dot{h}\rangle=& \mathcal{R} \langle \pi^\ast f \rangle X\langle \dot{h} \rangle - \mathcal{R} \langle  \dot{h} \rangle X\langle \pi^\ast f \rangle \\
=& -\dot{f}\dot{h},
\end{split}
\end{equation*}
for any $f,h\in\Gamma\Lambda^0M$. \\

Suppose now that $\Phi$ is horizontal.
By \cref{Lem X wedge Phi} (since $\Phi$ is radial and simple) there exists some $X\in\Gamma TU$ such that
\begin{equation*}
\Phi= \mathcal{R}\wedge X.
\end{equation*}
By the horizontal condition
\begin{equation*}
\begin{split}
-\dot{f}\dot{h}= \Phi\langle \pi^\ast f,\dot{h}\rangle=& \mathcal{R} \langle \pi^\ast f \rangle X\langle \dot{h} \rangle - \mathcal{R} \langle  \dot{h} \rangle X\langle \pi^\ast f \rangle \\
=& -X\langle \pi^\ast f \rangle \dot{h}.
\end{split}
\end{equation*}
Hence $X \langle \pi^\ast f \rangle =\dot{f}$, that is $X$ is horizontal.
\end{proof}

\begin{lemma}\label{Lem X Hor}
Let $X\in\Gamma TU$ and $\Phi\in\Gamma\mathcal{B}^2(U)$ such that $\Phi=\mathcal{R}\wedge X$. If $X$ is horizontal and $\Phi$ is integrable then there exists $\alpha\in\Gamma\Lambda^0U$ such that 
\begin{equation}
[\mathcal{R},X]=X+\alpha\mathcal{R}.
\end{equation}
\end{lemma}

\begin{proof}
Since $X$ is horizontal, observe that for any $f\in\Gamma\Lambda^0M$ we have
\begin{equation*}
\begin{split}
[\mathcal{R},X]\langle \pi^\ast f\rangle=& \mathcal{R}\langle X\langle \pi^\ast f\rangle\rangle - X\langle \mathcal{R}\langle \pi^\ast f \rangle \rangle\\
=& \mathcal{R}\langle \dot{f}\rangle =\dot{f}=X\langle \pi^\ast f\rangle.
\end{split}
\end{equation*}
By the integrability of $\Phi$ we also have $[\mathcal{R},X]=\alpha\mathcal{R}+ \beta X $ so that we also have
\begin{equation*}
[\mathcal{R},X]\langle \pi^\ast f\rangle= \alpha\mathcal{R}\langle \pi^\ast f\rangle +\beta X\langle \pi^\ast f \rangle= \beta X\langle \pi^\ast f \rangle.
\end{equation*}
Combining the last two equations gives us that $\beta=1$.
Hence the result holds.
\end{proof}

\subsection{Vlasov Bivectors}\label{Subsect Vlasov Bivectors}

With this we have all the necessary ingredients to define a system in terms of a Vlasov bivector.
There are several equivalent properties which ensure a bivector contains sufficient structure to define a system.
Given a horizontal bivector $\Psi\in\Gamma\mathcal{B}^2_H(U)$, these three properties are, the existence of a Vlasov bivector $W\in\Gamma TU$ such that $\Psi=\mathcal{R}\wedge W$, integrability, and another property we call the radially cubic property.

\begin{definition}[Radially Cubic]\label{Def Radially Cubic}
A \texttt{}horizontal bivector $\Psi$ is \textit{radially cubic} if for any $f,h\in\Gamma\Lambda^0 M$, $\underline{u}\in U$, and $\lambda\neq 0$ we have
\begin{equation}
    \Psi\langle \dot{f},\dot{h} \rangle \vert_{\lambda \underline{u}}= \lambda^3 \Psi\langle \dot{f},\dot{h} \rangle \vert_{\underline{u}} .
\end{equation}
\end{definition}

\def\coord{{\textup{crd}}}

We see below in \cref{Thm Vlasov Bivector} that we can generate a Vlasov field $W$, using a kinematic indicator. A natural choice to use is a "lab time" kinematic indicator $\dot{f}$. However in general there is no $f\in\Gamma\Lambda^0M$ such that $\dot{f}\ne0$ on all of $U$. As an intermediate step, we can always define the following kinematic indicator using the coordinate system $(x^1,\ldots x^n)$ on $M$,
\begin{equation}
    F_\coord= \sum_{\mu=1}^n (\dot{x}^\mu)^2.
\label{F Coord}
\end{equation} 
It is clear this is not a physical kinematic indicator since it depends on the coordinate system. 
It is homogeneous degree 2. Since this is even, one would need the causality indicator to define a kinematic domain, \cref{Eqn E hat}. 

\begin{definition}[Coordinate based Vlasov field]
Let $(x^1,\ldots x^n)$ be a local coordinate system on $M$ and let $\Psi\in\Gamma\mathcal{B}_H^2(U)$, we define the \textit{coordinate based Vlasov field} by
\begin{equation}\label{Eqn W xi}
    W_\coord= \frac{\Psi\langle F_\coord,\bullet \rangle}{2F_\coord}.
\end{equation}
\end{definition}

\begin{lemma}\label{Lem Wf}
Let $\Psi\in\Gamma\mathcal{B}^2_H(U)$ be radially cubic and then $W_\coord$ as given by \cref{Eqn W xi} is a Vlasov field such that 
\begin{equation}
    \Psi=\mathcal{R}\wedge W_\coord.
\end{equation}
Furthermore, given another representation for $\Psi$, say $\Psi=\mathcal{R}\wedge X$ for some $X\in\Gamma TU$, then 
\begin{equation}\label{Eqn Wf to X}
W_\coord =X- \frac{X\langle F_\coord\rangle}{2F_\coord}\mathcal{R}.
\end{equation}
\end{lemma}

\begin{proof}
By \cref{Psi Structure Thm} there exists a horizontal vector field $X\in\Gamma TU$ such that $\Psi=\mathcal{R}\wedge X$.
Then $W_\coord$ and $X$ are related by
\begin{equation*}
\begin{split}
    W_\coord=& \frac{1}{2 F_\coord} \left( \mathcal{R}\langle F_\coord\rangle X - X\langle F_\coord\rangle \mathcal{R} \right)\\
    =& X-\frac{X\langle F_\coord\rangle}{2 F_\coord}\mathcal{R}.
\end{split}
\end{equation*}

To see that $\Psi=\mathcal{R}\wedge W_\coord$ first suppose that $\Psi=\mathcal{R}\wedge X$ for some $X\in\Gamma TU$ then observe that
\begin{equation*}
\begin{split}
\mathcal{R}\wedge \frac{\Psi\langle F_\coord,\bullet\rangle}{2 F_\coord}=& \mathcal{R}\wedge \bigg{(} \frac{\mathcal{R}\langle F_\coord\rangle}{2 F_\coord}X- \frac{X\langle F_\coord\rangle}{2 F_\coord}\mathcal{R} \bigg{)}\\
=&\mathcal{R}\wedge X =\Psi.
\end{split}
\end{equation*}
Define local coordinates $(x^\mu,\dot{x}^\mu)$.
To see that $W_\coord$ is horizontal observe that for any $f\in\Gamma\Lambda^0M$ we have
\begin{equation*}
    \begin{split}
        W_\coord\langle \pi^\ast f\rangle =& \frac{\Psi\langle F_\coord,\pi^\ast f \rangle}{2F_\coord}
        =\sum_\mu \frac{\Psi\langle (\dot{x}^\mu)^2,\pi^\ast f \rangle}{2F_\coord}\\
        =&\sum_\mu \frac{\dot{x}^\mu \Psi\langle \dot{x}^\mu,\pi^\ast f \rangle}{F_\coord}
        =\sum_\mu \frac{(\dot{x}^\mu)^2\dot{f}}{F_\coord}=\dot{f}.
    \end{split}
\end{equation*}
To see that $W_\coord$ is radially quadratic observe that for any $\underline{u}\in U,\; \lambda\in\mathbb{R},\; f\in\Gamma\Lambda^0M$,
\begin{equation*}
    \begin{split}
        W_\coord\vert_{\lambda\underline{u}} \langle \dot{f}\rangle  =& \frac{\Psi\vert_{\lambda\underline{u}}\langle F_\coord,\dot{f} \rangle}{2F_\coord\vert_{\lambda\underline{u}}}\\
        =&\sum_\mu \left( \frac{\dot{x}^\mu}{F_\coord} \right) \bigg{\vert}_{\lambda\underline{u}} \Psi\vert_{\lambda\underline{u}}\langle \dot{x}^\mu,\dot{f} \rangle\\
        =&\frac{\lambda}{\lambda^2} \sum_\mu \left( \frac{\dot{x}^\mu}{F_\coord} \right) \bigg{\vert}_{\underline{u}} \lambda^3 \Psi\vert_{\underline{u}}\langle \dot{x}^\mu,\dot{f} \rangle\\
        =&\lambda^2 \frac{\Psi\vert_{\underline{u}}\langle (\dot{x}^\mu)^2,\dot{f}\rangle}{2F_\coord\vert_{\underline{u}}}
        =\lambda^2W_\coord\vert_{\underline{u}}\langle \dot{f}\rangle.
    \end{split}
\end{equation*}
Hence $W_\coord$ is a Vlasov field.
\end{proof}

\begin{theorem}\label{Thm Vlasov Bivector}
Let $\Psi\in\Gamma\mathcal{B}^2_H(U)$ be a horizontal bivector. The following properties are equivalent:
\begin{enumerate}
\item[\textnormal{(i)}] There exists a Vlasov field $W\in\Gamma TU$ such that $\Psi=\mathcal{R}\wedge W$.
\item[\textnormal{(ii)}] $\Psi$ is integrable (\cref{Def Integrability}).
\item[\textnormal{(iii)}] $\Psi$ is radially cubic (\cref{Def Radially Cubic}).
\end{enumerate}
\end{theorem}

\begin{proof}
Suppose first that there exists a Vlasov field $W$ such that $\Psi= \mathcal{R}\wedge W$.
To see that $\Psi$ is radially cubic observe that, for some $f,h\in\Gamma\Lambda^0M$, any $\underline{u}\in U$ and $\lambda\neq 0$, it follows that
\begin{equation*}
\begin{split}
\Psi\langle \dot{f},\dot{h} \rangle \vert_{\lambda \underline{u}}=& \mathcal{R} \langle \dot{f} \rangle W\langle \dot{h} \rangle \vert_{\lambda \underline{u}}- \mathcal{R} \langle \dot{h} \rangle W\langle \dot{f} \rangle \vert_{\lambda \underline{u}}\\
=& \lambda \mathcal{R} \langle \dot{f} \rangle \lambda^2 W\langle \dot{h} \rangle \vert_{\underline{u}} -\lambda \mathcal{R} \langle \dot{h} \rangle \lambda^2 W_i\langle \dot{f} \rangle \vert_{\underline{u}}\\
=& \lambda^3 \Psi\langle \dot{f},\dot{h} \rangle \vert_{\underline{u}}.
\end{split}
\end{equation*}

Also, since $W$ is radially quadratic we have $[\mathcal{R},W]=W$ by \cref{Lem VQ iff Com}.
Hence $\Psi$ is integrable.
That is, property (i) implies properties (ii) and (iii).

By \cref{Lem Wf}, if $\Psi$ is radially cubic then $W_\coord$ as defined by \cref{Eqn W xi} is a Vlasov field and $\Psi=\mathcal{R}\wedge W_\coord$. Hence property (iii) implies (i).

Suppose now that $\Psi$ is integrable.
Let $W_\coord$ be given by \cref{Eqn W xi}.
By \cref{Lem Wf}, $\Psi=\mathcal{R}\wedge W_\coord$ and $W_\coord$ is horizontal.
It remains to show that $W_\coord$ is radially quadratic.
By \cref{Psi Structure Thm} there exists a horizontal vector field $X\in\Gamma TU$ such that $\Psi=\mathcal{R}\wedge X$ and by \cref{Lem Wf}, $W_\coord$ and $X$ are related by \cref{Eqn Wf to X}.
The integrability of $\Psi$ tells us $[\mathcal{R},X]=\alpha \mathcal{R}+ \beta X$ for some $\alpha,\beta\in\Gamma\Lambda^0 U$.
It follows that $\beta=1$ form \cref{Lem X Hor}.
We then have that
\begin{equation*}
\begin{split}
[\mathcal{R},W_\coord ]=& [\mathcal{R},X]- \left[  \mathcal{R},\frac{X\langle \dot{f}\rangle}{\dot{f}} \mathcal{R} \right]\\
=& [ \mathcal{R},X] -\mathcal{R} \bigg{\langle} \frac{X\langle \dot{f}\rangle}{\dot{f}} \bigg{\rangle} \mathcal{R}\\
=& X+ \alpha\mathcal{R} -\left( \frac{\mathcal{R}\langle X\langle \dot{f}\rangle \rangle}{\dot{f}}- \frac{X\langle \dot{f}\rangle}{\dot{f}} \right) \mathcal{R}\\
=& X+\alpha \mathcal{R} -\frac{1}{\dot{f}} [\mathcal{R},X]\langle \dot{f}\rangle \mathcal{R}\\
=& X+\alpha \mathcal{R} -\frac{1}{\dot{f}} (X\langle \dot{f}\rangle +\alpha\dot{f})\mathcal{R}\\
=&X-\frac{X\langle \dot{f}\rangle}{\dot{f}}\mathcal{R} = W_\coord .
\end{split}
\end{equation*}
Hence $W_\coord$ is radially quadratic by \cref{Lem VQ iff Com}.
It follows that $W_\coord$ is a Vlasov field and property (ii) implies property (i).
\end{proof}

\begin{definition}[Vlasov Bivectors]\label{Def Bivector System}
A horizontal bivector $\Psi\in\Gamma\mathcal{B}^2_H(U)$ defines a \textit{Vlasov bivector} if it satisfies one of the following equivalent properties:
\begin{enumerate}
\item There exists a Vlasov field $W\in\Gamma TU$ such that $\Psi=\mathcal{R}\wedge W$.
\item $\Psi$ is integrable (\cref{Def Integrability}).
\item $\Psi$ is radially cubic (\cref{Def Radially Cubic}).
\end{enumerate}
The space of Vlasov vectors over $U$ is denoted by $\Psi\in\Gamma\mathcal{B}^2_V(U)$.
\end{definition}

\begin{lemma}
    Let $\Psi\in \Gamma\mathcal{B}_V^2(U)$, let $E$ have a kinematic indicator $F$ of order $k$, and let
\begin{equation}
{W}_F=\frac{\Psi \left \langle {F},\bullet \right\rangle}{k F},
\label{Eqn W from Psi}
\end{equation}
Then $W_F$ is a Vlasov field which is compatible with $F$ and
\begin{equation}
    \Psi = \mathcal{R}\wedge W_F
\end{equation}
\end{lemma}

\begin{proof}
$\Psi=\mathcal{R}\wedge W_F$ and $W_F$ is horizontal by the same logic as \cref{Lem Wf}.
Furthermore, by \cref{Lem Wf}, $\Psi=\mathcal{R}\wedge W_\coord$ where $W_\coord$ is given by \cref{Eqn W xi}.
To see that $W_F$ is radially quadratic observe that
\begin{equation*}
    \begin{split}
        W_F\vert_{\lambda\underline{u}} \langle \dot{f}\rangle =& \left( \frac{\Psi\langle F,\dot{f}\rangle}{kF} \right) \bigg{\vert}_{\lambda\underline{u}}\\
        =& \left( \frac{\mathcal{R}\langle F\rangle W_\coord\langle \dot{f}\rangle - \mathcal{R}\langle \dot{f}\rangle W_\coord}{kF} \right) \bigg{\vert}_{\lambda \underline{u}}\\
        =& \left( W_\coord\langle\dot{f}\rangle -\frac{\dot{f}}{kF} W_\coord\langle F\rangle  \right) \bigg{\vert}_{\lambda\underline{u}}\\
        =& \left(\lambda^2 W_\coord\langle\dot{f}\rangle -\frac{\lambda \dot{f}}{\lambda^k kF} \lambda^{k+1} W_\coord\langle F\rangle  \right) \bigg{\vert}_{\underline{u}}\\
        =& \lambda^2 \left(W_\coord\langle\dot{f}\rangle -\frac{ W_\coord\langle F\rangle }{ kF} \mathcal{R}\langle \dot{f} \rangle  \right) \bigg{\vert}_{\underline{u}}\\
        =&\lambda^2 W_F\vert_{\underline{u}}\langle\dot{f}\rangle.
    \end{split}
\end{equation*}

$F$ is compatible with $W_F$, i.e. $W_F\langle F\rangle=0$, since $\Psi\langle F,F\rangle =0$.
It follows that
\begin{equation*}
    \begin{split}
        \frac{\Psi\langle F,\bullet \rangle}{kF}=&  \frac{\Psi\langle F,\bullet \rangle}{\mathcal{R}\langle F\rangle}= \frac{\mathcal{R}\langle F\rangle W_F- W_F\langle F\rangle\mathcal{R}}{\mathcal{R}\langle F\rangle}= W_F.
    \end{split}
\end{equation*}
\end{proof}

With \cref{Lem Lin Rel Lemma} and \cref{Lem Lin Rel == Integrable} combined we may note that Vlasov bivectors constructed from projectively related Vlasov fields represents the same object (that is, if there exists a 1--homogeneous function $k\in\Gamma\Lambda^0 U$ such that $\hat{W}=W+k\mathcal{R}$, then $\Psi=\mathcal{R}\wedge W= \mathcal{R} \wedge \hat{W}$). In figure \cref{Geometric Bivector}, the foliations generated by $\mathcal{R},W$ are the same of those generated by $\mathcal{R},\hat{W}$.
We can use the Vlasov bivector to derive the transformation formula \cref{Magic Formula}.
\begin{lemma}\label{Lem MAgic Fromula From Psi}
    If $\Psi=\mathcal{R}\wedge W= \mathcal{R} \wedge \hat{W}$ and $\hat{W}$ is compatible with kinematic indicators $\hat{F}$ then \cref{Magic Formula} holds. 
\end{lemma}
\begin{proof}
Consider the action of $\Psi$ on the scalar field $\hat{F}$:
\begin{equation*}
    \begin{split}
        \Psi\langle \hat{F},\bullet \rangle =& \mathcal{R}\langle \hat{F} \rangle W- W\langle \hat{F}\rangle \mathcal{R}\\
        =& \mathcal{R}\langle \hat{F} \rangle \hat{W}- \hat{W}\langle \hat{F}\rangle \mathcal{R}= \mathcal{R}\langle \hat{F} \rangle \hat{W}.
    \end{split}
\end{equation*}
That is,
\begin{equation*}
     \hat{W}= W- \frac{W\langle \hat{F}\rangle}{\mathcal{R}\langle \hat{F}\rangle}\mathcal{R}.
\end{equation*}
\end{proof}

Here we show that the Vlasov Bivector $\Psi$, knit together as leaves, as depicted in \cref{Geometric Bivector}. Since the leaves are tangent to $\mathcal{R}$, then they must open out like a book. 
\begin{theorem}
    Given a Vlasov bivector $\Psi$, then for each $\underline{u}\in U$ there exists a 2--dimensions surface $K\subset U$ such that $\underline{u}\in U$ and $\Psi$ is a tangent to $K$. 
\end{theorem}
\begin{proof}
    From theorem \cref{Thm Vlasov Bivector} we can write $\Psi=\mathcal{R}\wedge W$ such that $\mathcal{R}$ and $W$ are in involution, \cref{Involution}. The result now follows directly from Frobenius Theorem, which can be found in \cite{abraham_foundations_2008}.
\end{proof}


\section{Particle Density \texorpdfstring{$(2n{-}2)$}{TEXT}--forms}\label{Sect Particle Density}

In this section we introduce the particle density $(2n{-}2)$--form on $U$, $\theta\in\Gamma\Lambda^{2n-2}U$. This may be depicted pictorially in \cref{fig:intro Leaves} and \cref{Geometric Bivector}.
It is the generalisation of the particle density form $\theta_E\in\Gamma\Lambda^{2n-2}E$. It is subject to the transport equations on $U$. 
This reformulation of the transport equations has the same advantages as mentioned in the introduction.

A particle density form on $U$ must satisfy certain conditions recover a particle density on $E$(and vice versa).
These are discussed in \cref{Subsec PDF on U into E}.
Recall the Null condition \cref{Def Null Condition}, which is needed to formulate the transport equations on $U$.
The null condition can be considered a generalisation of the $i_{W_E}\theta_E=0$ equation.

\subsection{Transport Equations on \texorpdfstring{$U$}{TEXT}}\label{Subsect Transport Equations on U}

\begin{definition}[Transport Equations on $U$]\label{Def Transport Equations on U}
Given a particle density form $\theta\in\Gamma\Lambda^{2n-2}U$ and a Vlasov bivector $\Psi\in\Gamma\mathcal{B}^2_V(U)$, the \textit{transport equations on} $U$ are given by
\begin{equation}\label{PF Transport Equations}
\mbox{Null}(\Psi,\theta)\;  \mbox{ holds and } \; d\theta=0,
\end{equation}
where the null condition is given by \cref{Def Null Condition}.
\end{definition}

Notice that if $\mbox{Null}(\Psi,\theta)$ holds and $\Psi=\mathcal{R}\wedge W$ then both $W$ and $\mathcal{R}$ are tangent to the form manifolds of $\theta$.
That is, $i_W\theta=0$ and $i_\mathcal{R}\theta=0$. 
This is consistent with the visualisation in \cref{Geometric Bivector}.
We can identify the form manifolds of $\theta$ with the leaves of a foliation generated by a Vlasov bivector.
Since $\theta$ is a closed form, the form manifolds associated with it are smooth surfaces.
Furthermore, the velocity density profile of the particle distribution is reflected in these form manifolds: the closer together the surfaces, the greater the local velocity density of the particles (and vice versa). 

\begin{definition}[Populated Systems on $U$]\label{Def Populate PF System}
Given a Vlasov bivector $\Psi$ and a particle density $(2n{-}2)$--form $\theta\in\Gamma\Lambda^{2n-2} U$ which satisfies the transport equations on $U$ \cref{PF Transport Equations}, we define the pair $(\Psi,\theta)$ to be a populated system on $U$.
\end{definition}

As stated in the introduction, defining a populated system does not require a time orientation.
Since for any given radial $\{\lambda\underline{u}\in U, \lambda>0\}$ is disconnected from $\{\lambda\underline{u}\in U, \lambda<0\}$ there is no relationship between $\theta\vert_{\underline{u}}$ and $\theta\vert_{-\underline{u}}$. 
For example, in \cref{Vlasov Bivector} the particle density form can be considered nonzero on $U^+$ and zero on $U^-$.

\subsection{Relating the particle density on \texorpdfstring{$U$}{TEXT} with the particle density on \texorpdfstring{$E$}{TEXT}}\label{Subsec PDF on U into E}

For this subsection we assume that $U$ is time orientable and that there is a kinematic domain $E$. We relate the particle densities $\theta\in\Gamma\Lambda^{n-2}U$ and $\theta_E\in\Gamma\Lambda^{n-2}E$. Furthermore we assume that $\theta$ only contains particles which lie on $U^+$.

\begin{definition}[Future Pointing Particle Density]
    A particle density form $\theta\in\Gamma\Lambda^{(2n-2)}U$ is called future pointing if $U$ is time orientable and
    \begin{equation}\label{Eqn Future Pointing 1}
        \theta\vert_{U^-}=0.
    \end{equation}
    For future pointing particle densities we let the restriction $\theta^+\in\Gamma\Lambda^{2n-2}U^+$ be
    \begin{equation}\label{Eqn Future Pointing 2}
        \theta^+ = \theta\vert_{U^+}.
    \end{equation}    
    A populated system $(\Psi,\theta^+)$ formed from a future pointing particle density is called a future pointing populated system.
\end{definition}

\begin{lemma}\label{Lem PF System to P System}
Let $(\Psi,\theta^+)$ define a future pointing populated system on $U$, satisfying the transport equations.
Given a kinematic domain $E\subset U^+$ with \Fkin $F$, let $W_E$ be given by \cref{Eqn Vlasov Pushforward} where $W$ is given by \cref{Eqn W from Psi}, and let $\theta_E\in\Gamma\Lambda^{2n-2}E$ be given by
\begin{equation}\label{Eqn Theta E}
\theta_E=\Sigma_E^\ast\theta^+.    
\end{equation}
Then $W_E$ and $\theta_E$ satisfy the transport equations on $E$.
\end{lemma}

\begin{proof}
Observe that
\begin{align*}
d\theta_E=& d\Sigma_E^\ast \theta= \Sigma_E^\ast d\theta=0, \; \mbox{ and }\;\\
i_{W_E}\theta_E=& i_{W_E} \Sigma_E^\ast \theta= \Sigma_E^\ast \left( i_W\theta \right) =0.
\end{align*}
Hence $\theta_E$ with $W_E$ satisfy \cref{Def Transport Equations E}.
\end{proof}

\begin{lemma}\label{Lem Theta from Theta E}
Given a kinematic domain $E$ with $\theta_E\in\Gamma\Lambda^{2n-2}E$ and $W_E\in\Gamma TE$ such that $\theta_E$ satisfies the transport equations on $E$, define the map,
\begin{align}
    \Pi \colon U^+\rightarrow E; & \quad \Pi(\lambda\underline{v})=\underline{v} \;\mbox{\textnormal{ where }}\; \underline{v}\in E \textup{ and } \lambda>0.
\end{align}
Let 
\begin{equation}
    \theta^+= \Pi^\ast \theta_E,
    \label{Eqn Theta on U+}
\end{equation}
let $\theta$ be future pointing, given by \cref{Eqn Future Pointing 1} and \cref{Eqn Future Pointing 2},
and let $W\in\Gamma TU$ be given by \cref{Eqn Vlasov Pushforward}.
The system $(\Psi=\mathcal{R}\wedge W,\theta)$ satisfies the transport equations on $U$.
\end{lemma}

\begin{proof}
Let $F$ be the 1--homogeneous \Fkin for $E\subset U^+$. Define a foliation on $U^+$ of kinematic domains $E_\ell$ for $\ell\in\mathbb{R}^+$ where $E_1=E$ and $F|_{E_\ell}=\ell$. 
We may define a coordinate system $(x^\mu,\ell,\xi^a)$ for $U^+$ where $\ell$ is constant on each $E_\ell$.
We may choose $\xi^a$ such that they are 0--homogeneous, and $\ell$ is 1--homogeneous by the 1--homogeneity of $F$.
By \cref{Lem Homogeneity of W<G>} we have for each $\lambda>0$ and $\underline{u}\in U^+$,
\begin{equation}\label{Eqn W Lambda}
    W\vert_{\lambda\underline{u}}\langle x^\mu\rangle= \lambda W\vert_{\underline{u}}\langle x^\mu\rangle, \;
    W\vert_{\lambda\underline{u}}\langle \xi^a\rangle= \lambda W\vert_{\underline{u}}\langle \xi^a\rangle.
\end{equation}
The map 
\begin{equation*}
    \Xi_\ell \colon E\rightarrow E_\ell \; ; \;\Xi_\ell(\underline{v})=\ell\underline{v},
\end{equation*}
is well defined since $E_\ell=\{ \ell\underline{v},\; \underline{v}\in E\} $.
Furthermore, it satisfies
\begin{align*}
    \Pi\circ\Sigma_{E_\ell}\circ \Xi_\ell= \mathbb{1}_E \; \mbox{ i.e. } \;
    \Xi_\ell=\left( \Pi\circ \Sigma_{E_\ell} \right)^{-1}.
\end{align*}
Observe that we have
\begin{equation*}
    \begin{split}
        {}&\Xi_\ell^\ast x^\mu= x^\mu, \quad \Xi_\ell^\ast \xi^a=\xi^a,\\
        {}&\Xi_\ell^\ast dx^\mu=dx^\mu, \quad \Xi_\ell^\ast d\xi^a=d\xi^a.
    \end{split}
\end{equation*}
Lastly, define $\Omega\in \Gamma\Lambda^{2n-1}U$ by
\begin{equation*}
    \Omega=dx^0\wedge\cdots \wedge dx^{n-1}\wedge d\xi^1\wedge\cdots\wedge d\xi^{n-1}.
\end{equation*}

First observe that $d\theta^+=0$ since the exterior derivative commutes with the pullback.
Notice also that $i_{\mathcal{R}}\theta^+=0$ since $\Pi_\ast\mathcal{R}=0$ we have
\begin{equation*}
    i_{\mathcal{R}}\theta^+= i_{\mathcal{R}} \Pi^\ast \theta_E= \Pi^\ast \left( i_{\Pi_\ast\mathcal{R}}\theta_E\right) =0.
\end{equation*}
To see that $i_W\theta^+=0$ define coordinate functions $y^0,\ldots,y^{2n-2}$ as $y^k=x^k$ for $0\leq k\leq n-1$, $y^k=v^{k-n+1}$ for $n\leq k \leq 2n-2$. 
Hence $(y^0,\ldots,y^{2n-2})$ is a coordinate system for $E$ while 
$(\ell,y^0,\ldots,y^{2n-2})$ is a coordinate system for $U^+$. With appropriate domains $y^k = \Pi^\ast y^k$ and $y^k = \Sigma_E^\ast y^k$.
Hence for $\Sigma_E\circ\Pi:U^+\to U^+$ we have $(\Sigma_E\circ\Pi)^\ast(y^k)=y^k$. 
Thus $\Omega=dy^0\wedge\cdots\wedge dy^{2n-2}$. 
Also,
\begin{equation*}
    (\Sigma_E\circ\Pi)^\ast f = f
    \qquad\text{and}\qquad
    (\Sigma_E\circ\Pi)^\ast i^{(y)}_k \Omega  = i^{(y)}_k \Omega 
\end{equation*}
where $f$ is a scalar which such that $\mathcal{R}\langle f\rangle=0$, so it is a function only of $(y^0,\ldots,y^{2n-2})$.

Since we have, for any $\underline{v}\in E$, 
\begin{equation*}
    W\vert_{\lambda\underline{v}} \langle y^k\rangle = \lambda W\vert_{\underline{v}} \langle y^k\rangle
\end{equation*}
by \cref{Eqn W Lambda}. Thus for $\underline{v}\in E$ (and $\lambda\underline{v}\in E_\lambda$ we have 
\begin{equation*}
    \begin{split}
    \Big(\Xi_\lambda^\ast \left( i_{W_{E_\lambda}} dy^k\right)\Big) \Big\vert_{\underline{v}}
    =& \; \Xi_\lambda^\ast \Big(\left( i_{W_{E_\lambda}} dy^k\right) \vert_{\lambda\underline{v}}\Big)\\
    =&i_{W_{E_\lambda}} dy^k\vert_{\lambda\underline{v}}\\
    =& \lambda i_{W_E} dy^k\vert_{\underline{v}} \\
    =& \lambda i_{W_E} \Xi_\lambda^\ast \left( dy^k\vert_{\lambda\underline{v}} \right)\\
    =& \Big(\lambda i_{W_E} \Xi_\lambda^\ast \left( dy^k\right)\Big)\Big\vert_{\underline{v}} .
    \end{split}
\end{equation*}
Since this is true for any $\underline{v}\in E$ we have (for $dy^k\in\Gamma\Lambda^1 E_\lambda$)
\begin{equation*}
    \Xi_\lambda^\ast \left( i_{W_{E_\lambda}} dy^k\right)= \lambda i_{W_E}\Xi_\lambda^\ast dy^k.
\end{equation*}
It follows that 
\begin{equation*}
    \lambda i_{W_{E}} \Xi_\lambda^\ast \Sigma_{E_\lambda}^\ast \Omega = \Xi_\lambda^\ast \left( i_{W_{E_\lambda}} \Sigma_{E_\lambda}^\ast\Omega \right).
\end{equation*}
Since $\theta_E$ is an $(2n-2)$--form on a $(2n-1)$--dimensional manifold $E$, with coordinates $(y^0,\ldots,y^{2n-2})$ we may write
\begin{equation*}
    \theta_E=\theta_E^k i_k^{(y)}\Sigma_E^\ast\Omega ,
\end{equation*}
where $\theta^\mu_E=\theta^\mu_E(y^0,\ldots,y^{2n-2})$. Thus we can extend $\theta^\mu_E$ to a scalar field on $U^+$ so that $\mathcal{R}\langle \theta_E^k\rangle=0$. Also $\theta^k_E=\Pi^\ast\theta_E^k=\Xi^\ast_\ell\theta_E^k$ for the appropriate domains.
\begin{equation*}
    \begin{split}
        \theta^+=& \Pi^\ast \theta_E\\
        =& \Pi^\ast \left( \theta_E^k i_k^{(y)}\Sigma_E^\ast\Omega  \right)\\
        =& \Pi^\ast\theta_E^k\;(\Sigma_E\circ\Pi)^\ast \left( i_k^{(y)}\Omega  \right)\\
        =& \theta_E^k\; i_k^{(y)}\Omega .
    \end{split}
\end{equation*}
We then have for each $\lambda$
\begin{equation*}
    \begin{split}
        \Xi_\lambda^\ast \left( \Sigma_{E_\lambda}^\ast \left( i_W\theta^+\right) \right) =&
        \Xi_\lambda^\ast \left( i_{W_{E_\lambda}} \Sigma_{E_\lambda}^\ast \Pi^\ast \theta_E \right)\\
        =& \; \Xi_\lambda^\ast \left( i_{W_{E_\lambda}} \Sigma_{E_\lambda}^\ast \Pi^\ast \left( \theta_E^k i_k^{(y)}\Sigma_E^\ast\Omega  \right) \right)\\
        =& -\theta_E^k i_k^{(y)} \Xi_\lambda^\ast \left( i_{W_{E_\lambda}} \Sigma_{E_\lambda}^\ast \Pi^\ast \Sigma_E^\ast\Omega   \right)\\
        =& -\theta_E^k i_k^{(y)} \Xi_\lambda^\ast \left( i_{W_{E_\lambda}} \Sigma_{E_\lambda}^\ast \Omega   \right)\\
        =& -\lambda \theta_E^k i_k^{(y)} i_{W_{E}} \Xi_\lambda^\ast \left( \Sigma_{E_\lambda}^\ast \Omega   \right)\\
        =&\; \lambda i_{W_E}\Xi_\lambda^\ast \left( \theta_E^k i_k^{(y)} \Sigma_{E_\lambda}^\ast \Omega   \right)\\
        =&\; \lambda i_{W_E}\Xi_\lambda^\ast \left(  \Sigma_{E_\lambda}^\ast \Pi^\ast \left( \theta_E^k i_k^{(y)}\Sigma_E^\ast\Omega  \right)\right)\\
        =&  \;\lambda  i_{W_{E}} \Xi_\lambda^\ast \left( \Sigma_{E_\lambda}^\ast \Pi^\ast \theta_E \right)\\
        =& \lambda\ i_{W_E} \theta_E =0
    \end{split}
\end{equation*}

Since $\Xi_\lambda$ is bijective we have $\Sigma_{E_\lambda}^\ast \left( i_W\theta^+ \right)=0$.
In order to see that this implies $i_W\theta^+=0$, observe the following result:
\begin{equation*}
    \begin{split}
    {}&\text{Given $\alpha\in\Gamma\Lambda^{2n-3} U^+$, if $i_{\mathcal{R}} \alpha=0$}\\
    {}&\text{and $\Sigma_{E_\lambda}^\ast \alpha=0$ for each $\lambda >0$, then $\alpha=0$.}
    \end{split}
\end{equation*}
This follows since $i_{\mathcal{R}} \alpha=0$, we may write $\alpha= \alpha^{kj}i_k^{(y)}i_j^{(y)} \Omega$.
For any $\lambda>0,$ $0=\Sigma_{E_\lambda}^\ast \alpha=\Sigma_{E_\lambda}^\ast \left( \alpha^{kj} \right) i_k^{(y)}i_j^{(y)} \Sigma_{E_\lambda}^\ast \Omega$.
Since the form on the RHS is nonzero we must have $\Sigma_{E_\lambda}^\ast \left( \alpha^{kj} \right)=0$ for all $\lambda>0$. Since $U^+$ is the union of $E_\lambda$ then $\alpha^{jk}=0$.
It follows that $(\Psi,\theta^+)$ satisfy the transport equations on $U^+$.
\end{proof}

\begin{lemma}
Given a kinematic domain $E$ and particle densities $\theta_E$ and $\theta$ where $\theta$ is future pointing satisfying the relevant transport equations, then \cref{Eqn Theta E} holds if and only if \cref{Eqn Theta on U+} holds.
\end{lemma}

\begin{proof}
First suppose that $\theta^+=\Pi^\ast \theta_E$.
We have
\begin{equation*}
    \Sigma_E^\ast \theta^+= \Sigma_E^\ast \Pi^\ast \theta_E= \left( \Pi \circ \Sigma_E \right)^\ast \theta_E= \theta_E.
\end{equation*}

Suppose we now have $\theta_E= \Sigma_E^\ast \theta^+$.
Define a coordinate system $(\ell,y^k)$  in the same way as in the proof for \cref{Lem Theta from Theta E}, so that $\left( \Sigma_E\circ \Pi \right)^\ast\! y^k=y^k$.
Since $d\theta^+=0$ and $i_\mathcal{R}\theta^+=0$ we have
\begin{equation*}
    \theta^+= \theta^k i_k^{(y)} \Omega_X, \quad \theta^+= \theta^+(y^0,...,y^{2n-2}),
\end{equation*}
where $\Omega_X=dy^0\wedge\cdots\wedge dy^{2n-2}$.
We then have
\begin{equation*}
    \begin{split}
        \Pi^\ast \theta_E=& \Pi^\ast\Sigma_E^\ast \theta^+
        = \Pi^\ast\Sigma_E^\ast \left( \theta^k i_k^{(y)} \Omega_X \right)\\
        =& \theta^k \left( \Sigma_E\circ \Pi \right)^\ast \left( i_k^{(y)} \Omega_X \right)\\
        =& \theta^k i_k^{(y)} \Omega_X
        = \theta^+.
    \end{split}
\end{equation*}
\end{proof}

\subsection{The current associated with the particle density form}
\label{Subsec Current Density}
For a populated system on $E$, $(W_E,\theta_E)$ we construct the current $J_E$ associated with $\theta_E$ by integrating $\theta_E$ over each fibre. In differential geometry, this integration can be expressed very naturally using the language of de Rham push forwards. The conservation of the charge $d J_E=0$ follows from the transport equation $d \theta_E=0$ and the fact the the exterior derivative commutes with the de Rham push forward. 

For the populated system on $U$, it is necessary introduce a support 1--form $\chi$, so as to define the current. It is then necessary to show that the result is independent of the choice of this support 1--form, and the charge is conserved.

\begin{definition}[Integration along a fibre]
Let $\pi\colon K\rightarrow M$ where $K$ is an oriented $k$--dimensional vector bundle over $M$.
A form $\alpha\in\Gamma\Lambda^qK$ (for $n\leq k \leq 2n$) is said to have vertical compact support if for each $p\in M$ the restriction $\alpha\vert_{\pi^{-1}(p)}$ has compact support.
Given a form with vertical compact support $\alpha\in\Gamma\Lambda^q K$ the \textit{integral along the fibre} (otherwise known as the deRham pushforward) $\pi_\varsigma\alpha\in\Gamma\Lambda^{n-r} M$ is defined by
\begin{equation}\label{Eqn deRham 2 Int}
\int_K \pi^\ast \beta \wedge \alpha= \int_M\beta\wedge \pi_\varsigma \alpha,
\end{equation}
for all forms $\beta\in\Gamma\Lambda^r M$ with compact support such that $q+r =\dim (K)$.
For a comprehensive overview see \cite{derham_differentiable_2012}.
\end{definition}

\begin{definition}[Current forms from $E$]\label{Def Current from E}
Let $U$ be time orientable, let $E\subset U^+$ be a kinematic domain and let $\theta_E\in\Gamma\Lambda^{2n-2}E$ define a particle density 6-form satisfying the transport equations on $E$. 
We define the current $(n{-}1)$--form $\mathcal{J}_E\in\Gamma\Lambda^{n-1}M$ by
\begin{equation}
\mathcal{J}_E= \pi_{E\varsigma}(\theta_E). \label{Current 3--form}
\end{equation}
Here, $\pi_E\colon E\rightarrow M$ is the projection from the bundle $E$ to $M$.
\end{definition}

Since the above definition relies on a choice of kinetic domain $E$, we propose the following generalisation.
For this we need to define a support form, which is a 1--form on $U$.

\begin{definition}[Support Form]
    Given any $\underline{u}\in U$ let $\mathfrak{R}_{\underline{u}}=\{ \lambda\underline{u} \; \colon \;  \lambda>0\}$ and $\hat{\mathfrak{R}}_{\underline{u}} \colon \mathbb{R}^+ \hookrightarrow U $.
    A \textit{support form} $\chi\in\Gamma\Lambda^1 U$ is a 1--form such that for all $\underline{u}\in U$, $\hat{\mathfrak{R}}_{\underline{u}}^\ast\chi$ has compact support on $\mathfrak{R}_{\underline{u}}$ and satisfies
    \begin{equation}\label{Eqn Support Form Condition}
        \int_{\mathbb{R}^+} \hat{\mathfrak{R}}_{\underline{u}}^\ast \chi =1,
    \end{equation}
    for each $\underline{u}\in U$.
\end{definition}

Observe that replacing $\underline{u}$ with $\lambda\underline{u}$ for some $\lambda>0$ does not change the integration in \cref{Eqn Support Form Condition}. 
Recall that $\theta\vert_{\underline{u}}$ may be unrelated to $\theta\vert_{-\underline{u}}$ and in general both can be non-zero. Thus we need a support form to have support on both $\mathfrak{R}_{\underline{u}}$ and $\mathfrak{R}_{-\underline{u}}$ so that both sides contribute.

\begin{definition}[Current forms from $U$]\label{Def Current from U}
Let $\theta\in\Gamma\Lambda^{2n-2}U$ be a particle density form satisfying the transport equations on $U$ (\cref{Def Transport Equations on U}).
The \textit{current form on $U$}, $\mathcal{J}\in\Gamma\Lambda^{n-1}M$ is given by
\begin{equation}
\mathcal{J}=\pi_\varsigma(\chi\wedge\theta),
\end{equation}
for any support form $\chi\in\Gamma\Lambda^1U$.
\end{definition}

Although we need to impose additional structure through the inclusion of the support form, we may show that the current form from $U$ is independent of our choice of support form.

\begin{lemma}\label{Lem Integration}
Let $N$ be an $\ell$ dimensional manifold, $\alpha\in\Gamma\Lambda^{\ell-1}N$, $\beta\in\Gamma\Lambda^1N$, $t\in\Gamma\Lambda^0N$ and $X\in\Gamma TN$ such that 
\begin{equation}
L_X\alpha=0,\quad i_X\alpha=0,\quad \mbox{and} \quad X\langle t\rangle=1.
\end{equation}
Let $K_{t_0}=\{ p\in N \colon t\vert_p=t_0 \}$ for some value $t_0$ and define the embedding $\Sigma_{t_0}\colon K_{t_0}\hookrightarrow N$.
Lastly, let $\eta_p(t)$ denote the integral curve of $X$ passing through $p$.
Then
\begin{equation}
\int_N \alpha\wedge \beta= \int_{p\in K_{t_0}} \Sigma_{t_0}^\ast\alpha \left( \int_{\mathbb{R}} \eta_{p}^\ast \beta \right),
\end{equation}
provided $\eta_{p}^\ast\beta$ has compact support on the domain of each $\eta_p$.
\end{lemma}

\begin{proof}
Define a coordinate system on $N$ by $(x^0,...,x^{\ell-2},t)$ such that $X=\partial_t$.
Since $i_X\alpha=0$ it follows that $\alpha=\alpha_0 dx^0\wedge \cdots \wedge dx^{\ell-2}$ for some $\alpha_0\in\Gamma\Lambda^0N$.
Let $\beta_0=i_X\beta$ then we have 
\begin{equation*}
\begin{split}
\int_N&\alpha\wedge \beta
\\
=& \int_N \alpha_0  dx^0\wedge \cdots \wedge dx^{\ell-2} \wedge (i_X\beta)dt\\
=& \int_{(x^0,...,x^{\ell-2})\in K_{t_0}} \int_{t\in\mathbb{R}} \alpha_0(x^0,...,x^{\ell-2}) dx^0\wedge\cdots\\
{}&\qquad\cdots \wedge dx^{\ell-2} \wedge \beta_0(x^0,...,x^{\ell-2},t)dt\\
=& \int_{(x^0,...,x^{\ell-2})\in K_{t_0}} \alpha_0(x^0,...,x^{\ell-2}) dx^0\wedge \cdots \wedge dx^{\ell-2}\\
{}& \; \int_{t\in\mathbb{R}} \beta_0(x^0,...,x^{\ell-2},t)dt\\
=& \int_{(x^0,...,x^{\ell-2})\in K_{t_0}} \Sigma_{t_0}^\ast \alpha \int_{t\in\mathbb{R}}\beta\vert_{(x^0,...,x^{\ell-2},t)}\\
=& \int_{(x^0,...,x^{\ell-2})\in K_{t_0}} \Sigma_{t_0}^\ast\alpha \left( \int_{\mathbb{R}} \eta_{(x^0,...,x^{\ell-2},t)}^\ast \beta \right).
 \end{split}
\end{equation*}
\end{proof}

\begin{lemma}\label{Lem PF Current 3--form}
Let the current $(n{-}2)$--form $\mathcal{J}\in\Gamma\Lambda^{n-1}M$ as given by \cref{Def Current from U}.
$\mathcal{J}$ is independent of the choice of support form $\chi\in\Gamma\Lambda^1U$.
\end{lemma}

\begin{proof}
Observe that since $\theta$ satisfies the transport equations (\cref{Def Null Condition}) we have $i_{\mathcal{R}}\theta=0$ and $d\theta=0$ so $L_\mathcal{R}\theta=0$.
It follows that
\begin{equation*}
L_\mathcal{R}(\pi^\ast \phi \wedge \theta)= \pi^\ast(L_{\pi_\ast\mathcal{R}}\phi)\wedge \theta -\pi^\ast\phi \wedge L_\mathcal{R} \theta=0,
\end{equation*}
since $\pi_\ast\mathcal{R}=0$.
We also have $i_\mathcal{R}(\pi^\ast \phi \wedge \theta)=0$ so we may apply \cref{Lem Integration} in the following way.
Let $\eta_{\underline{u}}$ be an integral curve of $\mathcal{R}$ passing through $\underline{u}\in U$ and let $K_{r_0}=\{ \underline{u}\in U \colon r\vert_{\underline{u}}=r_0 \}$ with $\Sigma_{r_0}\colon K_{r_0} \hookrightarrow U $ for some value $r_0$.
We then have
\begin{equation*}
\begin{split}
    \int_U \pi^\ast \phi \wedge \theta\wedge \chi =&
    \int_{\underline{u}\in K_{r_0}} \Sigma_{r_0}^\ast (\pi^\ast\phi\wedge\theta) \left( \int_{\mathbb{R}} \eta_{\underline{u}}^\ast \chi \right)\\
    =&\int_{K_{r_0}} \Sigma_{r_0}^\ast (\pi^\ast\phi\wedge\theta).
\end{split}
\end{equation*}
Hence $\mathcal{J}$ is independent of the choice of support form $\chi$.
\end{proof}

\begin{lemma}
The current $(n{-}2)$--form on $U$, $\mathcal{J}=\pi_\varsigma(\chi\wedge \theta)$, satisfies the continuity equation
\begin{equation}
d\mathcal{J}=0.
\end{equation}
\end{lemma}

\begin{proof}
The exterior derivative commutes with the deRham pushforward and $d\theta=0$ so we have
\begin{equation*}
d\mathcal{J}= d\pi_\varsigma(\chi\wedge\theta)=  \pi_\varsigma(d\chi\wedge \theta).
\end{equation*}
Since $\mathcal{J}$ is independent of our choice of $\chi$ by \cref{Lem PF Current 3--form} it suffices to pick $\chi$ such that $d\chi=0$.
By picking $r\in\Gamma\Lambda^0 U$ such that $dr\neq 0$ and $\mathcal{R}\langle r \rangle=1$ we may define a coordinate system $(x^\mu,r,y^{a})$.
By choosing
\begin{equation*}
\chi=\chi_r(r)dr,
\end{equation*}
where $\chi_r(r)$ is a function in $r$ with compact support (to comply with \cref{Eqn Support Form Condition}), we have $d\chi=0$.
\end{proof}

\begin{lemma}\label{Lem J eq JE}
Let $U$ be time orientable, $E$ be a kinematic domain and $\theta$ is future time pointing. Let $\theta$ and $\theta_E$ be related by \cref{Eqn Theta E} and \cref{Eqn Theta on U+}.
The current form from $U$, $\mathcal{J}\in\Gamma\Lambda^{n-1}M$ (\cref{Def Current from U}) and the current form from $E$, $\mathcal{J}_E\in\Gamma\Lambda^{n-1}M$ (\cref{Def Current from E}), are identical.
\end{lemma}

\begin{proof}
Let $F$ be the 1--homogeneous \Fkin associated with $E$, define
\begin{equation*}
r=\log F,
\end{equation*}
$K_{r_0}=\{ \underline{u}\in U \colon r\vert_{\underline{u}}=r_0 \}$, and $\Sigma_{r_0}\colon K_{r_0}\hookrightarrow U$.
First observe that $E=K_0$ and $\mathcal{R}\langle r\rangle=1$.
By application of \cref{Lem Integration} we have for any test form $\phi\in\Gamma_0\Lambda^1M$,
\begin{equation*}
\begin{split}
\int_U \pi^\ast \phi\wedge \theta \wedge \chi=& \int_{p\in K_0} \Sigma_0^\ast(\pi^\ast \phi\wedge\theta) \left( \int_{\mathbb{R}} \eta_p^\ast \chi \right)\\
=& \int_{ K_0} \Sigma_0^\ast(\pi^\ast \phi\wedge\theta)\\
=& \int_{ E} \Sigma_E^\ast(\pi^\ast \phi\wedge\theta)\\
=& \int_E \pi_E^\ast\phi \wedge \theta_E.
\end{split}
\end{equation*}
Hence for any $E$, $\mathcal{J}=\mathcal{J}_E$.
\end{proof}

\subsection{Discussion about the Stress-Energy 3--form}\label{subsect Stress energy}
Given a kinematic domain $E$ and $\alpha\in\Gamma\Lambda^1M$ the stress-energy $(n{-}1)$--form can be expressed
\begin{equation}\label{Stress-Energy 3-form}
\tau_E{}_\alpha= \pi_E{}_\varsigma (\hat{\alpha} \theta_E),
\end{equation}
where $\hat{\alpha}\in\Gamma\Lambda^0U$ is such that for each $\underline{u}\in U$, $\hat{\alpha}\vert_{\underline{u}}=\alpha\colon\underline{u}$.
The stress-energy 3-form can be converted into the usual stress-energy tensor
\begin{equation}\label{Eqn EV Stress-Energy}
    T_E^{\mu\nu}=\int_E\dot{x}^\mu\dot{x}^\nu f_E \frac{\sqrt{-\det g}}{\dot{x}_0} d^3\dot{x},
\end{equation}
where $d^3\dot{x} =d\dot{x}^1\wedge d\dot{x}^2\wedge d\dot{x}^3$,
via
\begin{align}\label{Stress-Energy conversion}
    T_E^{\mu\nu} = \star (dx^\mu \wedge \tau_E{}_{dx^\nu})
\end{align}
where $\star$ is the Hodge dual. See \cref{Lem EV Stress-Energy} in the appendix.

This can be placed on the right hand side of Einstein's equations to complete the Einstein-Vlasov system. 
To see that the above expression yields the correct relationship for an Einstein-Vlasov system. 

Due to the similarities between the definitions of the current $(n{-}1)$--form (\cref{Current 3--form}) and the stress-energy $(n{-}1)$--form (\cref{Stress-Energy 3-form}), it may be tempting to try and define the stress-energy $(n{-}1)$--form on $U$ according to
\begin{equation}\label{Stress-Energy 3-form U}
\tau= \pi_\varsigma (\chi\wedge \hat{\alpha}\theta).
\end{equation}
Unfortunately, unlike for $\mathcal{J}$, the stress-energy $(n{-}1)$--form depends on the choice of $\chi$.
This is because $d\left( \hat{\alpha} \theta \right) \neq 0$, so that \cref{Lem J eq JE} does not apply.
It can also be shown that $\tau_E{}_\alpha$, \cref{Stress-Energy 3-form}, is dependent on the choice of $E$.

Since the usual stress-energy forms is $\tau_{\EH\alpha}$, this makes the stress-energy tensor a metric quantity. 
This  stress-energy form is needed for the Einstein-Vlasov system. It is still advantageous to use the formalism in this article. First it makes it clear how the Vlasov equations i.e. the Vlasov bivector and the transport equations depend on the metric. This may be less explicit in the usual treatment. This is especially relevant if one need to vary the metric, as is done in \cite{andersson_variational_2019}, for instance. 

Another use, of our approach, is when considering non-metric compatible connections. As stated the trajectories no longer remain on $\EH$, \cref{Fig Mass Shell Slip}. However using the the Vlasov bivector and the particle density form $\theta$, this is no longer a problem. After calculating $\theta$, one can choose $\EH$ to calculate $\tau_{\EH\alpha}$.

It is an open question to see if an object similar to the stress-energy form can be defined which does not depend on the choice of $E$ or $\chi$.

\section{Conclusion}\label{Sec Conclusion}

In this article we presented an alternative way of representing the Vlasov equation, which did not rely on a kinematic domain. This corresponds to not prescribing a parameterisation for the underlying 2nd order ODE. For this we had to use many new concepts not normally associated with the Vlasov equation: We replaced 1--dimensional prolongations of trajectories on $E$, with 2--dimensional leaves on $U$; the Vlasov vector field with a Vlasov bivector, and extended the particle density form from $E$ to $U$. We give the formula for the current form on $M$ which depends only on the particle density form. We also give the formula for the stress-energy forms, which depend both on the particle density form and the choice of kinematic domains, or alternatively the support form $\chi$. Using the stress-energy forms, we can investigate the full Einstein-Vlasov system. 
In order to do this we have summarised the results for the standard Vlasov equation and given the relationship to sprays and semi-sprays.  

This approach has a number of advantages detailed in the introduction \cref{Sec: Intro Advantages}. These include wider applicability such as not requiring time orientation, non-metric compatible connections and pre-metric electrodynamics. Another advantage is that the formula for going from one kinematic domain to another falls out trivially. 

There are a number of applications for this approach. 
One can investigate the Vlasov equation for lightlike particles. This will enable one to extend the ultra-relativistic approximation \cite{burton_asymptotic_2007} from fluids to the Vlasov equation. This is particularly relevant when considering particles in particle accelerators as well as astrophysical plasmas near black holes and neutron stars. 

There is clearly a deep relationship between the work presented here and the idea of trajectory and solutions of 2nd order ODEs which do not have a parameter prescribed. In a follow-up article, we show how to define such trajectories and their connection to the leaves of the Vlasov bivector.

Other directions one may consider are to look at the relationship of this work with jet bundles and Finsler geometry. One can also look at how to generalise the Vlasov equation and the Bolzmann equation. In this latter case we replace the transport equation \cref{PF Transport Equations}, with $\text{Null}(\Psi,\theta)$ and $d\pi_\varsigma\theta=0$.

In summary we argue that the Vlasov bivector is the fundamental object to describe kenetic system, since it is invariant under reparameterisation.

\begin{appendices}

\section{Appendices}
\subsection{Sprays and Semi-Sprays} \label{Spray Section}
Vlasov fields can be formulated in terms of sprays and semi-sprays. 
For a detailed discussion of sprays the reader is directed towards \cite{shen_differential_2001} and \cite{miron_geometry_2002}.
For our purposes we define a spray as follows.
\begin{definition}[Spray]
A \textit{spray} on a smooth manifold $N$ is a smooth vector field  $X\in \Gamma T(\breve{T}N)$ which is expressed in local adapted coordinates $(x^\mu,\dot{x}^\mu)$ as 
\begin{equation}
X=\dot{x}^\mu\frac{\partial\;}{\partial x^\mu} +X^\mu\frac{\partial\;}{\partial \dot{x}^\mu},
\end{equation}
where $X^\mu= X^\mu(x,\dot{x})$ are local functions on $\breve{T}N$ satisfying
\begin{equation}
X^\mu\vert_{\lambda \underline{u}}=\lambda^2 X^\mu\vert_{\underline{u}},\; \lambda>0,\; \underline{u}\in \breve{T}N.
\end{equation}
In the case where $\breve{T}N$ is a conic bundle this can be shown to be equivalent to \cref{Def Vlasov Field on U}.
\end{definition}

The trajectories of a spray are defined in the same way as described in \cref{Subsec ICs}, satisfying the equation
\begin{equation}\label{Spray SODE}
    \frac{d^2C^\mu}{dt^2}  = X^\mu\left( C^\mu(t),\frac{dC^\mu}{dt} \right).
\end{equation}

\begin{definition}[Projectively related Sprays]\label{Def Projectively Related}
Two sprays $X$ and $\hat{X}$ are \textit{projectively related} if they have the same trajectories as point sets. 
That is, if $C(t)$ is a trajectory of $X$, then there exists a reparameterisation $t=t(s)$ such that $C(s):=C(t(s))$ is a geodesic of $\hat{X}$ (and vice versa).
\end{definition}

\begin{lemma}\label{Projectively Related Lemma}
Two sprays $X$ and $\hat{X}$ are projectively related if and only if there exists a 1--homogeneous scalar field $k\in\Gamma\Lambda^0 (\breve{T}N)$ such that
\begin{equation}\label{PR Eqn}
\hat{X}= X+k \mathcal{R},
\end{equation}
where $\mathcal{R}$ is the radial vector field on $\breve{T}N$.
\end{lemma}

\begin{proof}
See Z. Shen, Differential Geometry of Spray and Finsler Spaces, pages 173-174 \cite{shen_differential_2001}.
\cref{Lem Trajectories}
\end{proof}

In the literature a semi-spray is defined similarly to a spray on $\breve{T}N$, only without the homogeneity property.
For our purposes however, it is productive to define sprays on some hypersurface $K\subset \breve{T}N$.
We restrict our attention to hypersurfaces defined similarly to lab time bundles.

\begin{definition}[Semi-Spray]
Given $s\in\Gamma\Lambda^0N$ let $K=\{ \underline{u}\in \breve{T}N : \dot{s}\vert_{\underline{u}}=1\}$.
A \textit{semi-spray} $X_K\in\Gamma \breve{T}K$ is a vector field given by
\begin{equation}\label{Eqn SS Coords}
X_K=\frac{\partial\;}{\partial s}+ \dot{x}^{a} \frac{\partial\;}{\partial x^{a}}+ X_K^{a}(s,x^{a},\dot{x}^{a})\frac{\partial\;}{\partial \dot{x}^{a}} .
\end{equation}
Unlike a spray, there are no homogeneity conditions on $X_K^{a}$.
The semi-spray $X_K$ corresponds to a set of ordinary differential equations locally expressed as
\begin{equation}\label{Semi-Spray SODE}
\frac{d^2f^{a}}{ds^2}= X_K^{a}\bigg{(} s,f^{a}, \frac{df^{a}}{ds} \bigg{)} .
\end{equation}
Similarly to the case with a spray, $f(s)$ is a solution to \cref{Semi-Spray SODE} if and only if its lift $\dot{f}(s)=(1,f^{a}(s), df^{a}(s)/ds)$ is an integral curve of $X_K$.
\end{definition}

Observe that in the case where $K=\ELab{s}\subset U$, we may identify \cref{Eqn SS Coords} with \cref{Eqn LT Vlasov Coord}, a Vlasov field on a lab time bundle (\cref{Eqn Intro ELab}).

Given such a semi-spray on $K$ we may construct a spray on $\breve{T}N$ and vice versa.
The full details of the lemma can be found in \cite{shen_differential_2001}.
An example of this lemma in action is given below.
If we are given a semi-spray determined by coefficients $X_K^{a}$ over $K$ (equipped with a choice of parameterisation $s$) then we may construct a spray over $\breve{T}N$ with the following coefficients
\begin{equation}\label{SS-S Construction}
\begin{cases}
{}& X^0(x,\dot{x})=0\\
{}& X^{a}(x,\dot{x})= \dot{x}^0\dot{x}^0 X_K^{a}(x^0,x^{a},\dot{x}^{a}/\dot{x}^0),
\end{cases}
\end{equation}
where $s=x^0$.
Note that under this construction, $X_K$ is induced by $X$.
This is an example of the quadratic extension described in \cref{Lem WE to W}.
The freedom to choose $X^0$ in \cref{SS-S Construction} roughly corresponds to the freedom to reparametrise the spray according to \cref{PR Eqn}.
In the instance where we restrict ourselves to a Vlasov field on a lab time then \cref{Eqn W from WE I} can be identified with \cref{SS-S Construction}.

\subsection{Auxiliary Proofs}\label{Sect Proofs}
\begin{lemma}\label{Lem WE Tangent}
Let $X\in\Gamma TN,\; f\in\Gamma \Lambda^0N$ such that $df\neq 0$, and $f^{-1}\{0\}=K\subset N$ with $\Sigma_K\colon K\hookrightarrow N$.
We have that $X\vert_K\langle f \rangle=0$ if and only if there exists a unique $Y\in\Gamma TK$ such that $X\vert_K= \Sigma_K{}_\ast Y$ (i.e. $X$ is tangent to $K$).
\end{lemma}

\begin{proof}
Suppose $X$ is tangent to $K$ then we have
\begin{equation*}
X\vert_K\langle f \rangle = \Sigma_K{}_\ast Y\langle f \rangle= Y\langle f \circ \Sigma_K \rangle= Y\langle 0 \rangle =0.
\end{equation*}
Suppose now that $X\vert_K\langle f \rangle=0$.
If a suitable vector field exists it is unique by the injectivity of $\Sigma_K$.
Let $\ell$ be the dimension of $N$.
Since $df\neq0$ there exists a local coordinate system $\{ x^0=f,x^1,..,x^{\ell-1} \}$ where we may express $X$ as $X=X^\mu\partial_\mu^{(x)}$, where $\mu=0,...,\ell-1$. 
Since $X\langle f\rangle =0$, we have $X^0=X\langle x^0\rangle =0$.
In this coordinate system we also have $\Sigma_K: (x^0,...,x^{\ell-1})\mapsto (0,x^1,...,x^{\ell-1})$.
This allows us to define a vector field $Y\in\Gamma TK$ where locally $Y=Y^a\partial_a^{(x)},\; Y^a=\Sigma_K^\ast X^a$, for $a=1,...,\ell-1$.
It follows that $X\vert_K =\Sigma_K{}_\ast Y$ and hence $X$ is tangent to $K$.
\end{proof}

\begin{lemma}\label{Lem Non Metric Vlasov}
    Given a spacetime manifold $M$ with a pseudo-Reimann metric $g$, there exists a Vlasov field  $W\in\Gamma TU$ constructed from a force equation involving a non-metric compatible connection $\hat{\nabla}$ i.e. in local coordinates $(x,\dot{x})$,
    \begin{equation}\label{Eqn Nonmetric Vlasov}
        \begin{split}
        {}&\hat{\nabla}_{\dot{C}}\dot{C}=
        \widetilde{i_{\dot{C}}\mathcal{F}},\\
        {}&W= \dot{x}^\mu\partial_\mu^{(x)}+ \left( g^{\mu\nu}\mathcal{F}_{\nu\rho}\dot{x}^\rho - \hat{\Gamma}^\mu_{\rho\nu} \dot{x}^\nu\dot{x}^\rho \right) \partial_\mu^{(\dot{x})},
        \end{split}
    \end{equation}
    where $\dot{C}=C_\ast(\partial_\tau)$ and $\tau$ is the chosen parameterisation of the trajectories,
    such that the integral curves of $W$ will not lie on $\EH$.
\end{lemma}

\begin{proof}
    Let $\nabla$ denote the Levi-Civita connection built from $g$.
    Since $\hat{\nabla}$ is non-metric compatible, it has non-vanishing non-metricity:
    \begin{equation*}
        \hat{Q}= \hat{\nabla} g \neq 0.
    \end{equation*}
    Letting $\FkinUH$ denote the \Fkin of $\EH$.
    We observe that $W$ is not tangent to $\EH$ since there exists $\dot{C}$ such that
    \begin{equation*}
        \begin{split}
            W\vert_{\dot{C}}\langle \FkinUH \rangle=& \dot{C}_\ast (\partial_\tau) \langle \FkinUH \rangle 
            = \partial_\tau\langle \FkinUH\circ \dot{C}\rangle\\
            =& \partial_\tau C^\ast\big{(}g(\dot{C}, \dot{C})\big{)}
            = \dot{C}\langle g(\dot{C},\dot{C}) \rangle\\ 
            =& \hat{\nabla}_{\dot{C}}\big{(} g(\dot{C},\dot{C})\big{)} 
            = \hat{Q}(\dot{C},\dot{C},\dot{C}) + g(\hat{\nabla}_{\dot{C}}\dot{C},\dot{C})\\
            =& \hat{Q}(\dot{C},\dot{C},\dot{C}) + g\left(\widetilde{i_{\dot
            C}\mathcal{F}},\dot{C}\right)\\
            =& \hat{Q}(\dot{C},\dot{C},\dot{C}) + i_{\dot{C}}i_{\dot{C}}\mathcal{F}\\
            =& \hat{Q}(\dot{C},\dot{C},\dot{C}) \neq0,
        \end{split}
    \end{equation*}
    where $\tau$ is a parameter and $\dot{C}={C_\ast(\partial_\tau)}$.
    Hence $W$ is not tangent to $\EH$ by \cref{Lem WE Tangent}, and consequently, its integral curves will not remain on $\EH$.
    Conversely, if we replace $\hat{\nabla}$ with the Levi-Civita  connection $\nabla$ in \cref{Eqn Nonmetric Vlasov}, then we may see that $W\vert_{\dot{C}}\langle \FkinUH \rangle =0$ since $\nabla g=0$.
    Hence $W$ is tangent to $\EH$ and consequently its integral curves are confined to $\EH$.
\end{proof}

\begin{lemma}\label{Lem Null Geo}
    Let $M$ be a Minkowsky spacetime manifold of dimension $2$ and let $ s \in\Gamma\Lambda^0M $ define a lab time function such that $\ELab{ s }$ defines a lab time bundle (see \cref{Eqn Intro ELab}).
    The null geodesics parametrised by the induced lab time coordinate in general do not satisfy the geodesic equation.
    Consequently, the prolongations of lab time parametrised curves must be expressed in terms of the pre-geodesic equation.
\end{lemma}

\begin{proof}
    Let $C\colon \mathcal{I}\hookrightarrow M$ be a trajectory parametrised by $\tau\in\Gamma\Lambda^0\mathcal{I}$.
    Then the prolongation satisfies 
    \begin{equation*}
        \dot{C}=C_\ast(\partial_\tau), \quad \dot{C}\langle  s  \rangle= \frac{d ( s \circ C)}{d \tau} =1.
    \end{equation*}
    Let $(t,x)$ define coordinates on $M$.
    Null trajectories in Minkowsky spacetime form straight lines, for example,
    \begin{equation*}
        (t\circ C)(\tau)=(x\circ C)(\tau).
    \end{equation*}
    The trajectory above satisfies
    \begin{equation*}
        \begin{split}
            \frac{\partial ( s \circ C)}{\partial \tau}=& \frac{\partial (x\circ C)}{\partial \tau} \frac{\partial  s }{\partial x}+ \frac{\partial (t\circ C)}{\partial \tau} \frac{\partial  s }{\partial t}\\
            =& \frac{\partial (x\circ C)}{\partial \tau} \left( \frac{\partial  s }{\partial x}+  \frac{\partial  s }{\partial t}\right)=1
        \end{split}
    \end{equation*}
    By defining $X=x\circ C$ and $T=t\circ C$ we can write
    \begin{equation*}
        \begin{split}
            {}& \dot{X}=  \left( \frac{\partial  s }{\partial x}+ \frac{\partial  s }{\partial t} \right)^{-1} \frac{\partial\;}{\partial x},\\
            {}& \dot{T}=  \left( \frac{\partial  s }{\partial x}+ \frac{\partial  s }{\partial t} \right)^{-1} \frac{\partial\;}{\partial t}.
        \end{split}
    \end{equation*}
    It follows that
    \begin{equation*}
        \begin{split}
            \nabla_{\dot{X}}\dot{X}\langle s \rangle= \frac{d\;}{d\tau}\left( \frac{\partial  s }{\partial x}+ \frac{\partial  s }{\partial t} \right)^{-1} \frac{\partial s}{\partial x} \neq 0,
        \end{split}
    \end{equation*}
    in general since there is no specific relationship between $s$ and $t,x$.
    The same result follows for $\dot{T}$.
    Hence the null trajectories parametrised by the lab time function cannot obey the geodesic equation, they must obey a pre-geodesic equation.
    
    If however we choose to parametrise the null geodesics by the Minkowsky coordinate $t$ then the above equation does satisfy $\nabla_{\dot{X}}\dot{X}=\nabla_{\dot{T}}\dot{T}=0$.
\end{proof}

\begin{lemma}
Two pairs of vector fields, $X_1,X_2\in\Gamma TN$ and $Y_1,Y_2\in\Gamma TN$, are specially related if and only if $X_1\wedge X_2= Y_1\wedge Y_2$. 
\end{lemma}
\begin{proof}\label{Prf 1}
Let  $X_1,X_2$ and $Y_1,Y_2$ be specially related. We then have
\begin{equation*}
\begin{split}
X_1\wedge X_2=& \left(\alpha Y_1 +\beta Y_2\right)\wedge \left(\gamma Y_1+ \delta Y_2\right)\\
=& \alpha\delta Y_1\wedge Y_2 + \beta\gamma Y_2\wedge Y_1\\
=& \left(\alpha\delta -\beta\gamma\right) Y_1\wedge Y_2\\
=& Y_1\wedge Y_2.
\end{split}
\end{equation*}

Suppose $X_1\wedge X_2\neq0$ then there exists a basis $\{X_1,...,X_n\}$ such that $X_i\wedge X_j\neq0$ for $i\neq j,\; i,j=1,...,n$.
We may then express 
\begin{equation*}
Y_1= \alpha^i X_i,\quad Y_2=\beta^i X_i,
\end{equation*}
for scalar fields $\alpha^i,\beta^i\in\Gamma\Lambda^0N$.

Hence, if we have $X_1\wedge X_2= Y_1\wedge Y_2$, we observe that
\begin{equation*}
\begin{split}
X_1\wedge X_2=& Y_1\wedge Y_2\\
=& \alpha^i X_i \wedge\beta^j X_j\\
=& \big{(}\alpha^1 \beta^2- \alpha^2 \beta^1\big{)} X_1\wedge X_2.
\end{split}
\end{equation*}

It follows that $\alpha^1\beta^2-\beta^1\alpha^2=1$.
It remains to be shown that the other coefficients ($(\alpha^i,\beta^i)$ for $i>2$) vanish. 
Since the coefficients of the terms $X_i\wedge X_j$ for $i$ or $j>2$ are vanishing, we have $\alpha^i\beta^j- \alpha^j \beta^i=0$ for $i$ or $j>2$.
We have the following set of equations:
\begin{align}
{}& \alpha^i \beta^j=\beta^i \alpha^j, \quad \mbox{ for } i,j>2 \label{condt 1},\\
{}& \alpha^1 \beta^j= \alpha^j \beta^1, \quad \mbox{ for } j>2  \label{condt 2},\\
{}& \alpha^2 \beta^j= \alpha^j \beta^2, \quad \mbox{ for } j>2  \label{condt 3}.
\end{align}

Define column vectors
\begin{equation*}
A=\begin{pmatrix}
\alpha^3\\
\vdots\\
\alpha^n
\end{pmatrix}, \quad
B=\begin{pmatrix}
\beta^3\\
\vdots\\
\beta^n
\end{pmatrix}.
\end{equation*}
These column vectors satisfy the following relations
\begin{align}
{}&\alpha^1B= \beta^1A, & \mbox{ by \cref{condt 2}}  \label{condt 5},\\
{}&\alpha^2B=\beta^2A, & \mbox{ by \cref{condt 3}}   \label{condt 6}.
\end{align}

By multiplying \cref{condt 5} by $\beta^2$ and then applying \cref{condt 6} we observe that
\begin{equation}
\begin{split}
{}& \beta^2 \alpha^1 B= \beta^2 \beta^1 A= \beta^1 (\beta^2A)= \beta^1(\alpha^2 B)\\
\implies & 0=(\beta^2\alpha^1-\beta^1\alpha^2)B=B.
\end{split}
\end{equation}
Since at least one of $\beta^1, \beta^2$ is non-zero we must also have $A=0$ by \cref{condt 5} or \cref{condt 6}. 
Hence $X_1,X_2$ and $Y_1,Y_2$ are specially related.
\end{proof}

\begin{lemma}
If a bivector $\Phi\in\Gamma\mathcal{B}^2(N)$ is integrable then for any representation $\Phi=X\wedge Y,\; X,Y\in\Gamma TN$, we have
\begin{equation}
[X,Y]=\alpha X+ \beta Y,
\end{equation}
for some $\alpha,\beta\in\Gamma\Lambda^0N$.
\end{lemma}

\begin{proof}\label{Prf 2}
Since $\Phi$ is integrable there exists $V,Z\in\Gamma TN$ such that $\Phi=V\wedge Z$ and $[V,Z]=\gamma V+\delta Z$.
Let $\Phi$ have another representation $\Phi=X\wedge Y$.
By \cref{Lem Lin Rel Lemma} there exist scalar field $\lambda,\sigma,\rho,\kappa\in\Gamma\Lambda^0N$ such that
\begin{equation*}
\begin{split}
X=& \lambda V+\sigma Z\\
Y=& \rho V+\kappa Z.
\end{split}
\end{equation*}
The Lie bracket of $X$ and $Y$ satisfies
\begin{equation*}
\begin{split}
[X,Y]=& [\lambda V+\sigma Z, \rho V+ \kappa Z]\\
=& [\lambda V, \rho V]+ [\lambda V,\kappa Z]+ [\sigma Z, \rho V]+ [\sigma Z, \kappa Z].
\end{split}
\end{equation*}
Consider a term containing $V$ and $Z$ e.g.
\begin{equation*}
\begin{split}
[\lambda V, \kappa Z]=& \lambda V\langle \kappa \rangle Z -\kappa Z\langle \lambda \rangle V+ \lambda\kappa [V,Z]\\
=& \big{(}\lambda\kappa \gamma -\kappa Z\langle \lambda \rangle\big{)} V+ \big{(}\lambda\kappa\delta + \lambda V\langle \kappa \rangle\big{)} Z.
\end{split}
\end{equation*}
Consider a term containing two of the same types e.g.
\begin{equation*}
\begin{split}
[\lambda V,\rho V]=& \lambda V\langle \rho \rangle V -\rho V\langle \lambda\rangle V +\lambda\rho [V,V]\\
=& \lambda V\langle \rho \rangle V -\rho V\langle \lambda\rangle V.
\end{split}
\end{equation*}
Hence we may write
\begin{equation*}
\begin{split}
[X,Y]=&\; \lambda V\langle \rho \rangle V -\rho V\langle \lambda\rangle V\\
{}& + \big{(}\lambda\kappa \gamma -\kappa Z\langle \lambda \rangle\big{)} V+ \big{(}\lambda\kappa\delta + \lambda V\langle \kappa \rangle\big{)} Z\\
{}& + \big{(}\sigma\rho \gamma -\rho V\langle \sigma \rangle\big{)} Z+ \big{(}\sigma\rho\delta + \sigma Z\langle \rho \rangle\big{)} V\\
{}& + \sigma Z\langle \kappa \rangle Z -\kappa Z\langle \sigma\rangle Z\\
=&\; \big{(} \lambda V\langle \rho \rangle  -\rho V\langle \lambda\rangle + \lambda\kappa \gamma\\
{}&-\kappa Z\langle \lambda \rangle +\sigma\rho\delta + \sigma Z\langle \rho \rangle\big{)} V\\
{}& + \big{(}\sigma Z\langle \kappa \rangle  -\kappa Z\langle \sigma\rangle + \lambda\kappa\delta \\
{}&+ \lambda V\langle \kappa \rangle + \sigma\rho \gamma -\rho V\langle \sigma \rangle\big{)} Z\\
=& \gamma' V+ \delta' Z
\end{split}
\end{equation*}

There exist $\lambda',\sigma',\rho',\kappa'$ such that
\begin{equation*}
\begin{split}
V=&\lambda' X+\sigma'Y\\
Z=&\rho' X+\kappa'Y.
\end{split}
\end{equation*}
These are guaranteed to exist since $\lambda\kappa-\sigma\rho=1$.
It follows that
\begin{equation*}
[X,Y]=\big{(}\gamma'\lambda'+\delta'\rho'\big{)}X+ \big{(}\gamma'\sigma'+\delta'\kappa'\big{)}Y.
\end{equation*}
\end{proof}

\begin{lemma}\label{Lem EV Stress-Energy}
Let $(M,g)$ be a spacetime manifold, let $E$ be a kinematic domain with a local coordinate system $(x^\mu,\dot{x}^a)$, let $\Omega_E$ be the volume form as given by \cref{Eqn Vol EH}, and let $\theta_E$ be a particle density form built from $\Omega_E$.
The stress-energy 3--form $\tau_E{}_\alpha$ \cref{Stress-Energy 3-form} is equivalent to the Vlasov-Einstein stress-energy tensor \cref{Eqn EV Stress-Energy}.
\end{lemma}

\begin{proof}
First observe that $\widehat{dx^\mu}\vert_{\underline{u}}=dx^\mu : \underline{u}= u^\mu=\dot{x}^\mu\vert_{\underline{u}}$.
We may write $\theta_E$ as
\begin{equation*}
    \theta_E= f_E i_{W_E}\Omega_E= f_E\dot{x}^\mu i_\mu^{(x)}\Omega_E+ f_E\varphi_E^a i_a^{(\dot{x})} \Omega_E,
\end{equation*}
where $\varphi_E^a= W_E \langle \dot{x}^a \rangle$.
Here we use the coordinate-based definition for the deRham pushforward as seen in \cite{bott_differential_1982}.
Since only forms with maximal degree in the fibre coordinates contribute to the deRham pushforward we have
\begin{equation*}
    \begin{split}
         (dx^\mu &\wedge \tau_E{}_{dx^\nu}) \vert_p
        \\ & 
        = dx^\mu\wedge\pi_\varsigma \big( \widehat{dx^\nu} \theta_E \big)\big\vert_p
        \\
        {}& = -\big(dx^\mu\wedge i_\rho^{(x)} d^4x\big)\big\vert_p \int_{E_p} \widehat{dx^\nu} \dot{x}^\rho f_E \frac{\det (g)}{\dot{x}_0} d^3\dot{x}\\
        {}& = -d^4x\vert_p \int_{E_p} \dot{x}^\mu\dot{x}^\nu f_E \frac{\det (g)}{\dot{x}_0} d^3\dot{x}.
    \end{split}
\end{equation*}
Since $\star 1= \sqrt{\det (g)} d^4x$ and $\star\star1=-1$ we have for each $p\in M$
\begin{equation*}
    \star \left( dx^\mu\wedge\tau_E{}_{dx^\nu}  \right) \big|_p = \int_{E_p} \dot{x}^\mu\dot{x}^\nu f_E \frac{\sqrt{\det (g)}}{\dot{x}_0} d^3\dot{x},
\end{equation*}
so that \cref{Eqn EV Stress-Energy} holds.
\end{proof}

\end{appendices}

\subsection*{Acknowledgements}
The authors are grateful for conversions with Dr. Alex Warrick.

\subsection*{Funding}
FG is grateful for the support of the Faculty of Science and Technology at Lancaster University.
JG is grateful for the support of STFC (the Cockcroft Institute ST/V001612/1).

\subsection*{Author contributions}
The writing of this article and the contents herein are credited to F. Gunneberg and J. Gratus with the exception of \cref{subsec Vlasov LF}, the calculations in which were performed by H. Stanfield, who also contributed to the preliminary work in \cref{Sec Transforming Between Domains}.

\bibliography{Vlasov_Paper_1_Bib}{}
\bibliographystyle{unsrt.bst}
\end{document}